\documentclass[onefignum,onetabnum]{siamonline250211}

\usepackage{lipsum}
\usepackage{amsfonts}
\usepackage{amsopn}
\usepackage{graphicx}
\usepackage{epstopdf}
\usepackage{algorithmic}
\ifpdf
  \DeclareGraphicsExtensions{.eps,.pdf,.png,.jpg}
\else
  \DeclareGraphicsExtensions{.eps}
\fi

\usepackage{bm}
\usepackage{color}
\usepackage{optidef}
\usepackage{booktabs}
\usepackage{subcaption}
\usepackage{siunitx}

\usepackage{enumitem}
\setlist[enumerate]{leftmargin=.5in}
\setlist[itemize]{leftmargin=.5in}

\newsiamremark{remark}{Remark}
\newsiamremark{hypothesis}{Hypothesis}
\crefname{hypothesis}{Hypothesis}{Hypotheses}
\newsiamthm{claim}{Claim}
\newsiamremark{fact}{Fact}
\crefname{fact}{Fact}{Facts}

\newsiamthm{assumption}{Assumption}
\newsiamthm{problem}{Problem}

\def\zero{\bm{0}}
\def\one{\bm{1}}

\def\a{\bm{a}}
\def\b{\bm{b}}
\def\c{\bm{c}}

\def\e{\bm{e}}

\def\h{\bm{h}}

\def\p{\bm{p}}
\def\q{\bm{q}}

\def\u{\bm{u}}
\def\v{\bm{v}}
\def\w{\bm{w}}
\def\x{\bm{x}}
\def\y{\bm{y}}

\def\gammab{\bm{\gamma}}

\def\IC{\mathcal{I}}
\def\JC{\mathcal{J}}
\def\KC{\mathcal{K}}
\def\FC{\mathcal{F}}
\def\SC{\mathcal{S}}
\def\TC{\mathcal{T}}
\def\WC{\mathcal{W}}
\def\HC{\mathcal{H}}

\def\Real{\mathbb{R}}
\def\trans{\top}
\def\HT{\mathsf{H}}
\def\opt{\mathrm{opt}}
\def\sol{\mathrm{sol}}
\def\add{\mathrm{add}}
\def\feas{\mathrm{feas}}
\def\ave{\mathrm{ave}}
\def\real{\mathrm{real}}
\def\ident{\mathrm{ident}}
\def\out{\mathrm{out}}

\DeclareMathOperator{\rank}{rank}
\DeclareMathOperator{\equivSym}{\Leftrightarrow}
\DeclareMathOperator{\cone}{cone}
\DeclareMathOperator{\conv}{conv}
\DeclareMathOperator{\dist}{dist}
\DeclareMathOperator{\MRSA}{MRSA}

\headers{Hyperspectral Image Data Reduction for Endmember Extraction}{T. Mizutani}

\title{
  Hyperspectral Image Data Reduction for Endmember Extraction
  \thanks{arXiv preprint, May 2026.
    \funding{This research was supported by the Japan Society for the Promotion of Science (JSPS KAKENHI Grant Number 23K11229).}}}

    \author{Tomohiko Mizutani\thanks{Department of Mathematical and Systems Engineering, Shizuoka University}
  (\email{mizutani.t@shizuoka.ac.jp}).}

\begin{document}

\maketitle

\begin{abstract}
    Endmember extraction from hyperspectral images aims to identify the spectral signatures
    of materials present in a scene. Recent studies have shown that self-dictionary methods
    can achieve high extraction accuracy; however, their high computational cost limits their
    applicability to large-scale hyperspectral images. Although several approaches have been
    proposed to mitigate this issue, it remains a major challenge. Motivated by this situation,
    this paper pursues a data reduction approach. Assuming that a hyperspectral image follows
    the linear mixing model with the pure-pixel assumption, we develop a data reduction technique to remove
    pixels corresponding to mixtures of multiple endmember signatures.
    We analyze the theoretical properties of this reduction step and show that it preserves pixels that lie close to the
    endmembers. Building on this result, we propose a data-reduced self-dictionary method that
    integrates the data reduction with a self-dictionary method based on a linear programming
    formulation. Numerical experiments demonstrate that the proposed method can substantially
    reduce the computational time of the original self-dictionary method without sacrificing
    endmember extraction accuracy.
\end{abstract}



\section{Introduction} \label{sec: introduction}
Endmember extraction from a hyperspectral image (HSI) is a fundamental task in utilizing hyperspectral imaging technology
for Earth's surface observation, with applications ranging from vegetation monitoring to mineral and resource exploration.
Numerous endmember extraction methods have been proposed to date; see \cite{Bio12, Ma14, Gil20} for surveys of recent advances.
Among these, self-dictionary methods have received growing attention because of their ability to achieve high extraction accuracy on HSIs.
However, a major drawback of self-dictionary methods is their high computational cost, which hinders their applicability to large-scale HSIs.

Self-dictionary methods formulate the endmember extraction task as a sparse optimization problem, treating the HSI data itself as a dictionary,
and then solve a convex relaxation of this problem.
Endmembers are identified using the optimal solution of the relaxed problem.
Although various self-dictionary methods have been developed in the literature \cite{Bit12, Elh12, Ess12, Gil18, Ngu22},
this paper mainly focuses on the method proposed by Bittorf et al.\ \cite{Bit12},
which employs a linear programming (LP) formulation as the convex relaxation.
For brevity, we refer to this LP-based self-dictionary method simply as the ``LP method'' throughout the paper.
Solving this LP relaxation is the primary computational bottleneck in applying the LP method to large-scale HSIs,
because the size of the LP problem grows quadratically with the number of pixels in the image.

There are two main approaches to addressing this computational issue:
(i) developing efficient algorithms for solving the convex relaxation problem, and
(ii) developing preprocessing techniques that reduce the size of the HSI data before applying self-dictionary methods.
There is a line of research exploring the first approach.
For example, an efficient algorithm based on the column generation framework was proposed in \cite{Miz25} for the LP method.
However, even with this improvement, processing the Urban HSI dataset, which is a widely used benchmark in the remote sensing community,
still takes 8 hours.
In light of this situation, this paper pursues the second approach.

\subsection*{Signal Model for Endmember Extraction from HSIs}
To describe our research question and contributions,
we begin by introducing the signal model for endmember extraction from HSIs.
In this paper, we adopt the linear mixing model under the pure-pixel assumption.
Consider an HSI consisting of $n$ pixels, acquired by a hyperspectral sensor with $d$ spectral bands.
We represent the HSI as a matrix $A \in \Real^{d \times n}$,
where the $(i,j)$th entry denotes the measurement of the $i$th spectral band at the $j$th pixel.
Thus, the $j$th column of $A$ corresponds to the observed spectral signature at the $j$th pixel.
We refer to such a matrix $A$ as an \emph{HSI matrix}.
Given an HSI matrix $A = [\a_1, \ldots, \a_n] \in \Real^{d \times n}$,
a \emph{linear mixing model}, abbreviated as LMM, assumes that the observed spectral signatures $\a_1, \ldots, \a_n$ are generated as
\begin{align*}
    \a_i = \sum_{j=1}^r h_{ji} \w_j + \v_i \quad \text{for} \  i = 1, \ldots, n,
\end{align*}
where each $\w_j$ satisfies $\w_j \ge \zero$ and each $h_{ji}$ satisfies $\sum_{j=1}^{r} h_{ji} = 1$
and $h_{ji} \ge 0$. 
We refer to $\w_j$ as an \emph{endmember signature},
$h_{ji}$ as the \emph{abundance fraction} of the $j$th endmember in the $i$th pixel,
and $\v_i$ as the \emph{noise} associated with the $i$th pixel.
An endmember signature represents the spectral signature of a distinct material present in the observed scene,
while an \emph{endmember} refers to the material itself.

To explain the pure-pixel assumption, we express the model above in matrix form as
\begin{align} \label{eq: LMM}
    A = WH + V,
\end{align}
where $W = [\w_1, \ldots, \w_r] \in \Real^{d \times r}, H \in \Real^{r \times n}$
with $(j,i)$th element $h_{ji}$, and $V = [\v_1, \ldots, \v_n] \in \Real^{d \times n}$.
Note that both $W$ and $H$ are nonnegative matrices, and, in particular,
the sum of the elements in each column of $H$ equals one.
We refer to $W$ as the \emph{endmember matrix} of $A$.
We say that the $j$th endmember has a \emph{pure pixel}
if there exists a column of $H$ equal to the $j$th unit vector.
The pure-pixel assumption asserts that each endmember has at least one such pure pixel.
Equivalently, $H$ can be written in the form
\begin{align} \label{eq: H under pure-pixel assumption}
    H = [I, \bar{H}] \Pi \in \Real^{r \times n},
\end{align}
where $I$ is the $r \times r$ identity matrix, $\Pi$  is an $n \times n$ permutation matrix,
and $\bar{H}$ is an $r \times (n-r)$ submatrix.
When the pure-pixel assumption holds, the LMM implies that,
in the absence of noise (i.e., $V = 0$), the endmember signatures $\w_1, \ldots, \w_r$ appear exactly as columns of $A$.

\subsection*{Motivation and Research Question}
We examine an HSI matrix $A$ from a geometric perspective,
assuming that it follows the LMM with the pure-pixel assumption, as given in \eqref{eq: LMM} and \eqref{eq: H under pure-pixel assumption}.
Let $\cone(A)$ denote the conical hull of the columns of $A$.
To simplify the discussion, we assume $V = 0$.
In this setting, the endmember signatures $\w_1, \ldots, \w_r$ correspond to the extreme rays of $\cone(A)$.
Thus, by removing the columns of $A$ that lie in its interior, the extreme rays remain, which correspond to the endmembers.
When $V \neq 0$, the resulting columns may instead be close to the endmembers.
Determining whether a column of $A$ lies in the interior of $\cone(A)$ can be formulated as a convex programming problem,
and the associated computational cost is low.
Accordingly,
if this data reduction step succeeds in discarding only irrelevant columns while retaining those close to the endmembers,
it can substantially reduce the computational time of self-dictionary methods without compromising their extraction performance.

Figure~\ref{fig: samson} illustrates this observation using the Samson HSI dataset,
which is described in detail in Section~\ref{subsec: objectives, implementation details, and datasets}.
Hereafter, let $A$ denote the HSI matrix of the Samson dataset.
We computed the intersection points of the columns of $A$ with a hyperplane
and then projected them onto a two-dimensional space.
The figure shows the projected points of all columns of $A$ (left)
and those obtained after removing the columns lying in the interior of $\cone(A)$ (right).
It can be seen that a large number of columns of $A$ are removed without discarding columns close to the endmembers.
Specifically, $A$ has 9025 columns,
but after removing the columns contained in the interior of $\cone(A)$, only 20 columns remain.
The details of how this figure was generated are provided in Section~\ref{subsubsec: comparison with existing methods}.

At this point, a natural question arises:
After removing the columns that lie in the interior of the conical hull of the HSI matrix,
how close are the remaining ones to the endmembers?
To the best of our knowledge, this research question has not been investigated from a theoretical perspective in the existing literature.
In the works of Esser et al.\ \cite{Ess12} and Gillis and Luce \cite{Gil18},
data reduction techniques based on clustering were employed within self-dictionary methods.
However, the theoretical properties of these reduction steps were not analyzed in their studies.
In this paper, we address this question through a theoretical analysis.
Building on the resulting insights,
we then develop a data reduction algorithm, which we integrate into a self-dictionary method.
\begin{figure}[h]
    \centering
    \begin{minipage}{0.48\linewidth}
        \centering
        \includegraphics[scale=0.6]{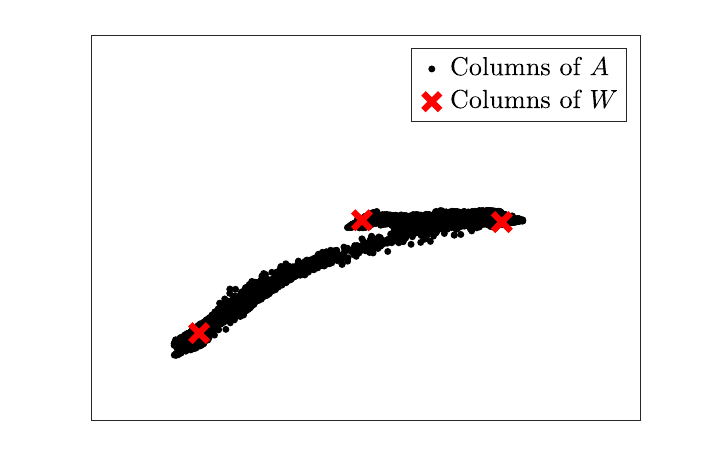}
    \end{minipage}
    \hfill
    \begin{minipage}{0.48\linewidth}
        \centering
        \includegraphics[scale=0.6]{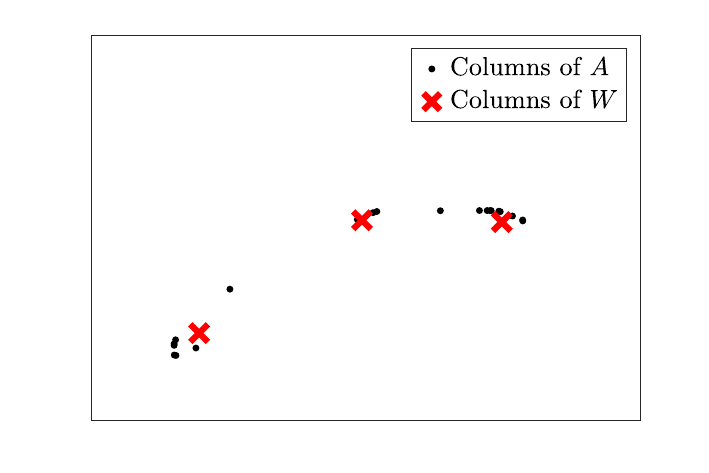}
    \end{minipage}
    \caption{Example of the Samson HSI dataset.
        The black dots represent the columns of the HSI matrix $A$,
        while the red x-markers represent the columns of the endmember matrix $W$.
        The left panel shows all columns of $A$,
        and the right panel shows the columns of $A$ after removing those lying in the interior of $\cone(A)$.}
    \label{fig: samson}
\end{figure}

\subsection*{Contributions}
The main contributions of this paper are threefold.
First, we address the research question posed above.
Our theoretical analysis is conducted under the assumption that the HSI matrix follows the LMM with the pure-pixel assumption.
We summarize our result as Theorem~\ref{thm: main result},
showing that even after removing the columns located inside the conical hull of the HSI matrix,
there still exist columns that lie close to the endmembers, provided that the noise level is below a certain threshold.

Second, we develop a data reduction algorithm based on the theoretical result,
and subsequently propose a data-reduced self-dictionary method, abbreviated as REDIC, for endmember extraction from HSIs
by integrating data reduction with the variant of the LP method proposed in \cite{Miz25}.
As observed in the case of the Samson dataset,
applying data reduction to HSI matrices can remove a large portion of their columns.
However, directly applying the LP method to the reduced data does not necessarily yield high extraction accuracy.
To address this issue,
we incorporate a data augmentation process together with an averaging technique to improve the extraction accuracy of REDIC.
For the reduced matrix obtained after data reduction,
we augment it by randomly selecting some of the removed columns and adding them back.
This augmentation step helps improve the accuracy of the extracted endmembers.
However, the randomness involved in the augmentation process introduces variability into the extraction results.
To mitigate this variability, we repeat the augmentation and endmember extraction using the LP method multiple times.
Finally, we average the outputs produced by the LP method to obtain the final estimates of the endmembers.

Third, we present numerical experiments demonstrating the performance of REDIC.
The experiments were conducted using both synthetic and real HSI datasets.
The results indicate that the data reduction algorithm within REDIC can remove a large portion of the columns in HSI matrices;
for example, for the Urban dataset, which originally contains about one hundred thousand columns,
only 483 columns remain after applying this data reduction.
In addition,
among the remaining columns, there exists at least one column close to each of the endmembers.
We compared the endmember extraction accuracy and computation time of REDIC with those of the LP method in \cite{Miz25}.
The results show that
REDIC can achieve extraction accuracy comparable to that of the LP method
while being substantially faster; for the Urban dataset,
REDIC requires only 16 minutes, whereas the LP method takes 8 hours.

The remainder of this paper is organized as follows.
Section \ref{sec: formulation of endmember extraction and self-dictionary approaches}
formulates the endmember extraction problem and reviews self-dictionary methods.
Section \ref{sec: theoretical result on data reduction and algorithm development }
presents our theoretical results on data reduction and describes the data reduction algorithm.
Section \ref{sec: REDIC method}
describes the REDIC method in detail and discusses related works on REDIC.
Section \ref{sec: experiments}
reports the results of numerical experiments.
Section \ref{sec: concluding remarks} concludes the paper and discusses future research directions.
Detailed proofs of the theoretical results are provided in
Appendices~\ref{sec: analysis of reduced hyperspectral image data} and~\ref{sec: analysis of DR and DRS}.

\subsection*{Notation}
We write $\Real^{d \times n}_+$ for the set of nonnegative matrices of size $d \times n$.
We use $\zero$ to denote the vector of all zeros, $\one$ the vector of all ones, $\e_i$ the $i$th unit vector,
$O$ the matrix of all zeros, $I$ the identity matrix, and $\Pi$ a permutation matrix.
For a positive integer $n$, we write $[n] = \{1, \ldots, n\}$.
For a set $\SC \subset [n]$ and an element $s \in \SC$,
we use $\SC - s$ to denote the set obtained by removing $s$ from $\SC$, i.e., $\SC \setminus \{s\}$.

Let $\a \in \Real^d$ and $A \in \Real^{d \times n}$.
The notation $\a(i)$ denotes the $i$th entry of $\a$, and $A(i,j)$ denotes the $(i,j)$th entry of $A$.
We denote by $A^\trans$ the transpose of $A$ and by $\rank(A)$ the rank of $A$.
Let $A = [\a_1, \ldots, \a_n]$.
For $\KC \subset [n]$,
we use $A(\KC)$ to denote the submatrix of $A$ obtained by retaining only the columns whose indices belong to $\KC$.
For example, if $\KC = \{1,3\}$, then $A(\KC) = [\a_1, \a_3]$.
The conical hull of the columns of $A$ is denoted by $\cone(A)$, i.e., $\cone(A) = \{ A \x \mid \x \ge \zero \}$.
We use $\| \cdot \|$ to denote a vector or matrix norm.
In particular, $\| \cdot \|_p$ denotes the $L_p$ norm and $\| \cdot \|_F$ denotes the Frobenius norm.
The $L_p$ norm for $A \in \Real^{d \times n}$ is defined by $\| A \|_p = \max_{ \| \x \|_p = 1} \ \| A \x \|_p$.
When $p = 1$, we have $ \| A \|_1 = \max_{j=1, \ldots, n} \sum_{i=1}^d |A(i,j)| $.

\section{Self-Dictionary Approaches to the Endmember Extraction Problem}
\label{sec: formulation of endmember extraction and self-dictionary approaches}

\subsection{Problem Formulation} \label{subsec: problem formulation}
We define the endmember extraction problem as described in Problem~\ref{prob: endmember extraction},
where we assume that the HSI matrix follows the LMM under the pure-pixel assumption
and that the number of endmembers is known in advance.

\begin{problem} \label{prob: endmember extraction}
Suppose that an HSI matrix $A \in \Real^{d \times n}$ is generated by the LMM $A = WH + V$, as given in~\eqref{eq: LMM},
where the pure-pixel assumption holds; that is, $H$ has the form~\eqref{eq: H under pure-pixel assumption}.
Given the HSI matrix $A$ and the number of endmembers $r$,
the objective is to estimate the $r$ columns of the endmember matrix $W$.
\end{problem}

In this paper, we use the mean-removed spectral angle (MRSA)
to evaluate the similarity between the true endmember signatures and their estimates.
For $\c \in \Real^d$, we define $\ave(\c) = (\one^\trans \c / d) \cdot \one$.
The MRSA value between the spectral signature vectors $\a, \b \in \Real^d$ is defined as
\begin{align*}
    \mathrm{MRSA}(\a, \b)
    = \frac{100}{\pi}
    \arccos
    \frac{
        (\a - \ave(\a))^\trans (\b - \ave(\b))
    }{
        \| \a - \ave(\a) \|_2 \, \| \b - \ave(\b) \|_2
    }
    \in [0,100].
\end{align*}
A smaller MRSA value for $\a$ and $\b$ indicates that $\a$ is more similar to $\b$.

Problem \ref{prob: endmember extraction} is essentially equivalent to computing nonnegative matrix factorizations (NMFs)
under the separability assumption.
In the context of NMF, $r$ is referred to as the factorization rank.
Considerable research has been devoted to the separable NMF problem; see \cite{Fu19, Gil20} for surveys of recent advances.
Arora et al.\ \cite{Aro12a} showed that $W$ can be recovered from $A$ in polynomial time when $V = O$.
Even when $V \neq O$, it is still possible to extract a submatrix of $A$ that is close to $W$,
provided that the noise $V$ is below a certain threshold.
Separable NMFs have applications not only to endmember extraction from HSIs,
but also to topic modeling in document corpora \cite{Aro12b, Aro13}.
Other applications include community detection \cite{Hua19}, spectral clustering \cite{Miz22b}, and crowdsourcing \cite{Ibr19}.

\subsection{Review of Self-Dictionary Methods}
Numerous endmember extraction methods have been proposed to date.
These methods can be broadly classified into two categories: greedy methods and convex programming-based methods.
Representative examples of greedy methods include PPI \cite{Boa95}, N-FINDR \cite{Win99}, SPA \cite{Ara01, Gil14a}, and VCA \cite{Nas05},
all of which are widely used for endmember extraction.
Greedy methods have low computational cost and can therefore handle large-scale datasets,
although the accuracy of the extracted endmembers may not always be sufficient.
In contrast, convex programming-based methods generally achieve higher extraction accuracy,
but their high computational cost remains a drawback.

We consider self-dictionary approaches within the category of convex programming-based methods.
Let $A \in \Real^{d \times n}$ denote an HSI matrix, and let $r$ denote the number of endmembers.
Regarding the matrix $A$ as a dictionary and its columns as atoms,
we interpret Problem~\ref{prob: endmember extraction} as the task of identifying $r$ atoms from this dictionary
such that the entire dictionary can be well approximated by a convex combination of the selected atoms.
Based on this interpretation, we formulate the problem as the following sparse optimization problem:
\begin{alignat}{4} \label{eq: self-dictionary formulation}
    \min_{X \in \Real^{n \times n}} & \quad & \| A - AX \|                                                             & \quad &
    \text{subject to}               & \quad & \| X \|_{\mathrm{row}, 0} = r, \ \one^\trans X = \one^\trans, \ X \ge O,
\end{alignat}
where $\| X \|_{\mathrm{row}, 0}$ denotes the number of nonzero rows of $X$.

Solving problem~\eqref{eq: self-dictionary formulation} exactly is difficult
because of the combinatorial constraint $\| X \|_{\mathrm{row}, 0} = r$.
Self-dictionary methods therefore typically rely on a convex relaxation of this problem.
As mentioned in Section~\ref{sec: introduction}, this paper focuses on the LP method proposed by Bittorf et al.\ \cite{Bit12},
which employs an LP formulation as the convex relaxation.
This method was originally developed in the context of identifying hot topics in document corpora and was therefore named Hottopixx.
Since its original proposal,
the LP method has been extended and refined by several researchers \cite{Gil13, Gil14b, Miz22a}.
In what follows,
we review the version of the LP method proposed in \cite{Miz22a}.
This method is based on the following optimization model:
\begin{alignat*}{5}
    \HT: & \quad & \min_{X \in \Real^{n \times n}} & \quad & \| A - AX \|_1                 &       &                     \\
         &       & \text{subject to}               &       & \sum_{i=1}^{n} X(i,i) = r,     &       &                     \\
         &       &                                 &       & 0 \le X(i,j) \le X(i,i) \le 1, & \quad & i,j = 1, \ldots, n.
\end{alignat*}

Model $\HT$ relaxes the combinatorial constraint $\| X \|_{\mathrm{row}, 0} = r$ in problem \eqref{eq: self-dictionary formulation}
and can therefore be reformulated as an LP.
Feasible solutions $X$ may have more than $r$ nonzero rows.
However, the number of such rows is expected to be small,
since the constraints restrict the number of nonzero diagonal elements in $X$ and, moreover,
if the $i$th diagonal element $X(i,i)$ is zero, then the $i$th row of $X$ must also be zero.

The LP method consists of two main steps:
it first computes the optimal solution $X_{\opt}$ of $\HT$ and then performs a postprocessing step based on $X_{\opt}$.
In the postprocessing step, $r$ clusters of the column indices of the input matrix $A$ are formed according to
the diagonal elements of $X_{\opt}$.
One representative element is then chosen from each cluster and returned as the output.
Theorem 3.2 of \cite{Miz22a} guarantees that
the corresponding columns of $A$ indexed by these elements are close to the $r$ endmember signatures,
provided that the noise level is below a certain threshold.

From a theoretical perspective, the LP method is attractive,
as the theoretical result suggests that the accuracy of the extracted endmembers can be high.
However, its major drawback lies in the high computational cost.
Specifically,
the method requires solving the model $\HT$, which reduces to an LP with $O(n^2)$ variables and $O(n^2)$ constraints;
see \cite{Miz25} for details.
Consequently, the method becomes impractical for HSIs with a large number of pixels $n$.
To address this computational issue, a row and column expansion (RCE) algorithm was proposed in \cite{Miz25}
to solve model $\HT$ efficiently. The RCE algorithm is based on a column generation technique,
a well-known approach for solving large-scale LPs.
Nevertheless, even with RCE, the improvement in computation time is insufficient.
For example,
the Urban dataset requires 8 hours of computation,
whereas greedy methods such as SPA and VCA finish in only a few seconds.
This observation motivates us to develop a data reduction technique for hyperspectral images.

In addition to the LP method, various self-dictionary approaches have been proposed
for endmember extraction from HSIs \cite{Elh12, Ess12, Gil18, Ngu22}.
FGNSR, proposed by Gillis and Luce \cite{Gil18}, and MERIT, proposed by Nguyen et al.\ \cite{Ngu22},
are most closely related to the LP method.
These methods are based on the following optimization model:
\begin{alignat}{4} \label{eq: self-dictionary formulation for FGNSR and MERIT}
    \min_{X \in \Real^{n \times n}} & \quad & \frac{1}{2} \| A - AX \|_F^2 + \lambda \cdot \Phi(X) & \quad & \mbox{subject to} & \quad & X \in \FC
\end{alignat}
where $\lambda$ is a penalty parameter and $\Phi$ is a regularization term designed to promote row sparsity in $X$.
FGNSR and MERIT identify endmembers by leveraging the optimal solution of this model.
As with the LP method, solving this model incurs a high computational cost and becomes impractical for HSIs with a large number of pixels.
The work by Esser et al.\ \cite{Ess12} can be viewed as a self-dictionary approach to the endmember extraction problem,
whereas the work by Elhamifar et al.\ \cite{Elh12} addresses a related problem within a similar modeling framework.

\section{Theoretical Result on Data Reduction and Algorithm Development}
\label{sec: theoretical result on data reduction and algorithm development }
\subsection{Main Theorem} \label{subsec: main theorem}
We now present the main theoretical result (Theorem \ref{thm: main result}),
which addresses the research question posed in Section~\ref{sec: introduction}.
Let an HSI matrix $A$ follow the LMM with the pure-pixel assumption,
as given in \eqref{eq: LMM} and \eqref{eq: H under pure-pixel assumption}.
As discussed in Section~\ref{sec: introduction},
when $V = O$, the endmember signatures $\w_1, \ldots, \w_r$ correspond to the extreme rays of $\cone(A)$.
Thus, by removing the columns of $A$ that lie in its interior, the extreme rays remain and correspond to the endmembers.
In contrast, when $V \neq O$, these endmember signatures may no longer be the extreme rays of $\cone(A)$.
However, the theorem shows that, even after removing the columns located inside $\cone(A)$,
there still exist columns that lie close to the endmembers, provided that the noise level is below a certain threshold.

We introduce the necessary terminology to describe the theoretical result.
Let $M \in \Real^{d \times n}_+$ be factored as
\begin{align} \label{eq: separable matrix}
    M = W H \in \Real^{d \times n}_+, \quad
    W \in \Real^{d \times r}_+, \quad H = [I, \bar{H}] \Pi \in \Real^{r \times n}_+.
\end{align}
We say that such a matrix $M$ is {\it $r$-separable}.
If an $r$-separable matrix $M$ is perturbed by noise so that $A = M + V$ for $V \in \Real^{d \times n}$,
we say that $A$ is {\it nearly $r$-separable}.
Accordingly, if an HSI matrix follows the LMM with the pure-pixel assumption, it is nearly $r$-separable.
Our analysis makes the following assumption regarding the nearly $r$-separable matrix $A = M + V$:
\begin{assumption} \label{asm: nearly r-separable matrix}
    Let $A = M + V \in \Real^{d \times n}$ be a nearly $r$-separable matrix,
    where $M \in \Real^{d \times n}_+$ is $r$-separable of the form $M = WH$, as shown in \eqref{eq: separable matrix},
    and $V \in \Real^{d \times n}$ is noise.
    We assume that the following conditions hold:
    \begin{enumerate}[label={\normalfont(\alph*)}]
        \item Each column of $W$ and $H$ has unit $L_1$ norm.
        \item $\| V \|_1 \le \epsilon$ for some $\epsilon \ge 0$.
    \end{enumerate}
\end{assumption}

A similar assumption has been used in analyses of the noise robustness of LP methods for separable NMFs \cite{Gil13, Miz22a}.
While \cite{Miz22a} assumes $0 \le \epsilon < 1$,
this condition is automatically satisfied in the setting of Theorem~\ref{thm: main result} due to the remaining assumptions.
Part (a) can be imposed without loss of generality for $r$-separable matrices.
Our analysis employs the conditioning parameter $\rho(W)$ for $W \in \Real^{d \times r}_+$, defined by
\begin{align*}
    \rho(W) = \min_{\| \x \|_1 = 1} \| W \x \|_1.
\end{align*}
This parameter measures the degree of linear independence among the columns of $W$.
In particular,
$\rho(W) > 0$ if and only if the $r$ columns of $W$ are linearly independent.
Moreover, if part (a) of Assumption~\ref{asm: nearly r-separable matrix} holds, then $\rho(W) \le 1$.
These properties are formally established in Lemmas~\ref{lem: linear independence and tau} and \ref{lem: upper bound on tau}.
Since our analysis relies on the linear independence of the columns of $W$, we use the parameter $\rho(W)$.
For $A \in \Real^{d \times n}$,
define the family $\Gamma(A)$ by
\begin{align*}
    \Gamma(A) = \{\KC \subset [n] \mid \cone(A) = \cone(A(\KC)) \}.
\end{align*}
Our theoretical result is summarized in the following theorem.
\begin{theorem} \label{thm: main result}
    Let $A \in \Real^{d \times n}$ be a nearly $r$-separable matrix of the form $A = M + V$,
    where $M = WH$ is $r$-separable and $V$ is noise.
    Assume that Assumption~\ref{asm: nearly r-separable matrix} holds for $A$.
    Let $\KC \in \Gamma(A)$.
    If $\epsilon < \rho(W) / 9$, then
    there exist $k_1, \ldots, k_r \in \KC$ such that
    \begin{itemize}
        \item $k_1, \ldots, k_r$ are all distinct, and
        \item $\| \w_j - \a_{k_j} \|_1 < (9 / \rho(W) + 1)\epsilon$ for each $j \in [r]$.
    \end{itemize}
\end{theorem}

The proof is given in Appendix~\ref{sec: analysis of reduced hyperspectral image data}.
This result shows that $\KC \in \Gamma(A)$ contains columns of $A$
whose errors relative to the endmember signatures are bounded by $O(\epsilon / \rho(W))$,
provided that $\rho(W)$ is not too small.
However, when $\rho(W)$ is too small relative to $\epsilon$, this result does not apply
and hence does not guarantee the existence of such columns.
For instance, this situation can occur when the columns of $W$ are nearly linearly dependent.

Theorem~\ref{thm: main result} is closely related to the noise robustness results established for LP methods.
Here, we briefly discuss the relationship between these results,
focusing on the result presented in \cite{Miz22a},
which is summarized in Theorem~3.2 of that paper and is essentially equivalent to Theorem~3.5 of \cite{Gil13}.
In what follows, we write $\rho$ for $\rho(W)$ for brevity.
The result of \cite{Miz22a} involves the following parameters, defined for $W = [\w_1, \ldots, \w_r] \in \Real^{d \times r}$:
\begin{align*}
    \kappa = \min_{1 \le j \le r} \min_{\x \ge \zero} \| \w_j - W([r] - j) \x \|_1,
    \quad
    \xi  = \min_{1 \le j_1 \neq j_2 \le r} \| \w_{j_1} - \w_{j_2} \|_1.
\end{align*}
In the original definition, the second parameter is denoted by $\omega$.
Since $\omega$ is used for a different purpose in this paper, we use $\xi$
instead to avoid confusion.

By arguments similar to those used in Lemma~3.8 of \cite{Gil14c} and Lemma~7.1 of \cite{Gil20},
we can show that $\rho \le \kappa$ for every $W \in \Real^{d \times r}$.
Using this relationship, Theorem~3.2 of \cite{Miz22a} implies that,
if the noise level $\epsilon$ is bounded by $O(\rho \xi / r)$,
then the LP method returns an estimated endmember matrix $W_{\out}$
such that the error with respect to the true endmember matrix $W$, namely $\| W - W_{\out} \|_1$,
is bounded by $O(r \epsilon / \rho)$.
Theorem~3.2 of \cite{Miz22a} requires Assumption~\ref{asm: nearly r-separable matrix}.
Hence, $\xi$ satisfies $0 \le \xi \le 2$; see Section~2 of \cite{Miz22a} for details.
From this discussion and Theorem~\ref{thm: main result},
we see that $\KC \in \Gamma(A)$ contains estimates of the endmembers
whose errors are smaller by a factor of $r$ than that of the estimates obtained by the LP method.
It should be noted that the robustness result for the LP method used in this discussion
is weaker than the original result in \cite{Miz22a}, since we use the inequality $\rho \le \kappa$.

\subsection{Data Reduction Algorithms} \label{subsec: data reduction algorithms}
We develop algorithms to find a set $\KC \in \Gamma(A)$ with a small number of elements.
Algorithm~\ref{alg: DR}, referred to as DR, describes a procedure for identifying a small set $\KC \in \Gamma(A)$.
DR generates a finite number of sets $\KC_1, \ldots, \KC_t$
that satisfy $\cone(A(\KC_1)) = \cdots = \cone(A(\KC_t))$.
In particular, since $\KC_1 = [n]$ and $\KC_t = \KC$, the final output $\KC$ belongs to $\Gamma(A)$.
Moreover, removing any element from $\KC$ yields a set that no longer belongs to $\Gamma(A)$.
Therefore, the output $\KC$ is a set with the smallest number of elements in $\Gamma(A)$.
A formal statement of this result is presented in Theorem~\ref{thm: analysis of DR and DRS outputs}
in Appendix~\ref{sec: analysis of DR and DRS}.

\begin{algorithm}[h]
    \caption{DR: Data reduction}
    \label{alg: DR}
    \smallskip
    Input: $A \in \Real^{d \times n}$ \\
    Output: $\KC \subset [n]$
    \smallskip
    \begin{enumerate}[label={\arabic*.}]
        \item Initialize $\KC_1 \gets [n]$ and $\ell \gets 1$.
        \item For $i = 1, \ldots, n$, do:
              if $\a_i \in \cone(A(\KC_\ell - i))$,
              then set $\KC_{\ell+1} \gets \KC_\ell - i$ and increment $\ell \gets \ell + 1$.
        \item Set $\KC \gets \KC_\ell$ and return $\KC$.
    \end{enumerate}
\end{algorithm}

To perform the feasibility test of whether $\a_i \in \cone(A(\KC_\ell - i))$,
we formulate it as a nonnegative least-squares problem:
\begin{alignat}{2} \label{eq: NNLS for feasibility test}
    \min_{\x \ge \zero} \quad \| A(\KC_\ell - i)\x - \a_i \|_2^2.
\end{alignat}
Efficient algorithms for solving this problem are available, such as the active set method and the projected gradient method.
The optimal value is zero if and only if $\a_i \in \cone(A(\KC_\ell - i))$.
We adopt this formulation for the feasibility test in Step~2.
Let $\x_{\sol}$ denote the numerically obtained solution to problem~\eqref{eq: NNLS for feasibility test}.
If $\x_{\sol}$ satisfies
\begin{align*}
    \| A(\KC_\ell - i)\x_{\sol} - \a_i \|_2 < \epsilon_{\feas},
\end{align*}
where $\epsilon_{\feas}$ is a small positive tolerance,
we regard $\a_i$ as belonging to $\cone(A(\KC_\ell - i))$ and remove $\a_i$ from $A$.

To enhance the computational efficiency of DR, we employ a data-splitting technique.
To simplify the description of this procedure, we introduce the following notation.
For a positive integer $m$,
let $\IC = \{i_1, \ldots, i_m\}$ be a set of $m$ distinct integers, and let $\JC \subset [m]$.
We use the notation $\IC(\JC)$ to denote the set $\{ i_j \mid j \in \JC \}$.
The DR algorithm with data splitting proceeds as follows.
First, we cluster the $n$ columns of $A$ into $p$ groups using the $k$-means method.
Let $\IC_1, \ldots, \IC_p$ denote the resulting $p$ groups of the column index set $[n]$.
These groups satisfy $\IC_1 \cup \cdots \cup \IC_p = [n]$ and $\IC_i \cap \IC_j = \emptyset$ for all distinct $i,j \in [p]$.
The collection $\{ \IC_1, \ldots, \IC_p \}$ is referred to as a \emph{$p$-way partition} of $[n]$.
Next, we apply DR to each submatrix $A(\IC_u)$ for $u \in [p]$
and obtain the corresponding outputs, denoted by $\JC_1, \ldots, \JC_p$.
For each $u \in [p]$, set $\KC_u = \IC_u(\JC_u)$.
We then construct $\IC = \KC_1 \cup \cdots \cup \KC_p$ and apply DR once more to the submatrix $A(\IC)$.
Let $\JC$ be the output, and set $\KC = \IC(\JC)$.
We return $\KC$ as the final output.
The overall procedure is summarized in Algorithm~\ref{alg: DRS}, which we refer to as DRS.
The parameter $p$ is a positive integer representing the number of clusters.
Similar to DR, the output $\KC$ of DRS is a set with the smallest number of elements in $\Gamma(A)$.
A formal statement of this result is presented in Theorem~\ref{thm: analysis of DR and DRS outputs} in Appendix~\ref{sec: analysis of DR and DRS}.

\begin{algorithm}[h]
    \caption{DRS: Data reduction via splitting}
    \label{alg: DRS}
    \smallskip
    Input: $A \in \Real^{d \times n}$; parameter $p$ \\
    Output: $\KC \subset [n]$
    \smallskip
    \begin{enumerate}[label={\arabic*.}]
        \item Apply the $k$-means method to the columns of $A$
              to construct the $p$-way partition $\{\IC_1, \ldots, \IC_p\}$ of the column index set $[n]$ of $A$.
        \item For $u = 1, \ldots, p$, apply DR to $A(\IC_u)$ to obtain the output $\JC_u$.
              Set $\KC_u = \IC_u(\JC_u)$.
        \item Let $\IC = \KC_1 \cup \cdots \cup \KC_p$.
              Apply DR to $A(\IC)$ to obtain the output $\JC$.
              Set $\KC = \IC(\JC)$ and return $\KC$.
    \end{enumerate}
\end{algorithm}

We discuss the computational cost of DRS in comparison with that of DR.
The input matrix of DR has $n$ columns,
whereas DRS operates on submatrices with fewer columns in each iteration.
Specifically, in Step~2 of DRS, the input matrix of DR has $|\IC_u|$ columns for each $u \in [p]$,
and in Step~3, it has $|\IC|$ columns.
Accordingly, the computational cost of solving problem~\eqref{eq: NNLS for feasibility test}
during the feasibility tests in DRS is smaller than that in DR.
On the other hand, the number of feasibility tests performed by DRS is larger than that of DR,
since DR performs $n$ tests, whereas DRS performs $\sum_{u=1}^p |\IC_u| + |\IC| = n + |\IC|$ tests.
We show in Section~\ref{subsubsec: computational time of DRS} that,
if the parameter $p$, which represents the number of clusters, is appropriately chosen,
the number of feasibility tests performed by DRS is nearly the same as that of DR.
Consequently, the overall computational time of DRS is substantially smaller than that of DR.

We conclude this section with a discussion of related works.
Finding a set $\KC \in \Gamma(A)$ with the smallest number of elements
is referred to as the conical hull problem in the paper by Dul\'{a} et al.\ \cite{Dul98}.
They proposed an efficient algorithm for solving this problem and reported numerical results.
Starting with $\KC = \emptyset$, their algorithm constructs $\KC$ by sequentially adding elements.
This contrasts with our DR algorithm, which begins with $\KC = [n]$, where $n$ is the number of columns of $A$,
and removes elements one by one.
At each iteration, their algorithm performs a feasibility test using an LP formulation,
whereas our DR algorithm employs a nonnegative least-squares formulation.
As mentioned in \cite{Fuk18}, their algorithm can be regarded as an instance of the algorithmic framework studied in \cite{Cla94}.

Assume that $A \in \Real^{d \times n}$ and $\u \in \Real^d$ satisfy $A^\trans \gammab > 0$ and $\u^\trans \gammab > 0$
for some $\gammab \in \Real^d$.
Then, checking whether $\u \in \cone(A)$ is equivalent to checking whether $\bar{\u} \in \conv(\bar{A})$
where $\bar{\u}$ is obtained by scaling $\u$, and $\bar{A}$ is obtained by scaling the columns of $A$.
Here, $\conv(\bar{A})$ denotes the convex hull of the columns of $\bar{A}$.
Such a vector $\gammab$ always exists when $\u > 0$ and $A > 0$;
hence, this assumption is likely to hold in Step~2 of DR when the algorithm is applied to HSI matrices.
Kalantari \cite{Kal15} developed a geometric algorithm called the Triangle Algorithm
for testing whether $\u \in \conv(A)$ and analyzed its computational complexity.

\section{The REDIC Method for Endmember Extraction} \label{sec: REDIC method}
\subsection{Algorithm Description} \label{subsec: algorithm description}
We develop REDIC, a data-reduced self-dictionary method
for endmember extraction from HSIs by combining DRS with the LP method of \cite{Miz22a}.
Here, we describe the details of REDIC.
To implement the LP method, REDIC employs the approach proposed in \cite{Miz25}.
This approach is referred to as EEHT (efficient and effective Hottopixx), where Hottopixx is the original name of the LP method.

First, EEHT applies a dimensionality reduction technique to the input HSI matrix.
Given $A \in \Real^{d \times n}$ and a positive integer $r$, we compute the top-$r$ truncated singular value decomposition (SVD) of $A$ as
\[
    A_r = U_r \Sigma_r V_r^\trans.
\]
Here, $\Sigma_r \in \Real^{r \times r}$ is a diagonal matrix whose entries are the top $r$ singular values of $A$,
and $U_r \in \Real^{d \times r}$ (resp.\ $V_r \in \Real^{n \times r}$) are the left (resp.\ right) singular vectors
associated with these singular values.
We then construct a reduced matrix $A' = \Sigma_r V_r^\trans \in \Real^{r \times n}$.
It should be noted that if $A$ follows the LMM, i.e., $A = WH + V$, then $A'$ also follows the LMM, i.e., $A' = W'H + V'$,
but with $r$ rows instead of $d$ rows.
Second, EEHT constructs the model $\HT$ using $A'$ instead of $A$ and solves it with the RCE algorithm proposed in \cite{Miz25}.
Finally, EEHT performs a postprocessing step:
it partitions the column index set of $A$ into $r$ clusters based on the optimal solution of $\HT$
and chooses elements from these clusters according to a prescribed criterion.
For a detailed description of EEHT, we refer readers to Algorithm~4 in \cite{Miz25}.
The experimental results demonstrate that RCE significantly reduces the computation time required to solve $\HT$.
Moreover, the accuracy of the extracted endmembers can be improved
when the postprocessing step chooses representative elements from clusters based on their centroids, a procedure referred to as method-C.
EEHT employing method-C for postprocessing is denoted as EEHT-C in \cite{Miz25}.

We design the framework of REDIC as follows.
First, it applies SVD-based dimensionality reduction and DRS to the input matrix as a preprocessing step
and then constructs a reduced matrix.
Second, it builds the model $\HT$ using the reduced matrix and solves it with RCE.
Finally, it performs the postprocessing using method-C.
As shown in Section~\ref{sec: experiments},
this preprocessing step allows for the elimination of a large number of columns from HSI matrices.
However, the performance of endmember extraction achieved by this framework is not always satisfactory.
Therefore, to improve its performance, we incorporate a data augmentation technique and an averaging technique into the framework.

Let $A'$ be the reduced matrix obtained by applying SVD-based dimensionality reduction to the input matrix $A$.
Let $\KC$ be the index set obtained after applying DRS to $A'$.
The data augmentation technique randomly selects a predetermined number of elements from the set $[n] \setminus \KC$
and adds them to $\KC$.
After this augmentation, endmember extraction is performed on the submatrix of $A'$ whose columns are indexed by the augmented set.
This technique can improve the performance of endmember extraction;
however, the results may vary due to the inherent randomness of the selection process.
To mitigate this variability, we employ an averaging strategy:
the endmember extraction process is repeated multiple times with different random selections,
and the resulting estimates are then averaged.

We summarize the overall procedure of REDIC in Algorithm~\ref{alg: REDIC}.
The parameters $r$ and $p$ are positive integers representing the estimated number of endmembers
and the number of clusters for DRS, respectively.
The parameters $\lambda$ and $\tau$ are nonnegative integers
representing the number of elements added in the data augmentation step and the number of repetitions for averaging, respectively.

\begin{algorithm}[h]
    \caption{REDIC: Data-reduced self-dictionary method for endmember extraction}
    \label{alg: REDIC}
    \smallskip
    Input: $A \in \Real^{d \times n}$; parameters $r$, $p$, $\lambda$, $\tau$ \\
    Output: $\hat{\w}_1, \ldots, \hat{\w}_r \in \Real^{d}$
    \smallskip
    \begin{enumerate}[label={\arabic*.}]
        \item Compute the top-$r$ truncated SVD $A_r = U_r \Sigma_r V_r^\trans$ of $A$
              and construct $A' = \Sigma_r V_r^\trans \in \Real^{r \times n}$.
        \item Apply DRS to $A'$ with parameter $p$ to obtain $\KC \in \Gamma(A')$.
        \item For $j = 1, \ldots, \tau$, do:
              \begin{enumerate}[label={3-\arabic*.}]
                  \item If $\lambda - |\KC| \le 0$, then  set $\KC_{\add} = \emptyset$;
                        otherwise, randomly select $\lambda - |\KC|$ elements from $[n] \setminus \KC$ and construct the set $\KC_{\add}$.
                  \item Construct the model $\HT$ using $A'(\KC \cup \KC_{\add})$ and $r$, solve it with RCE,
                        and perform the postprocessing step using method-C to obtain $r$ indices $i_1, \ldots, i_r$.
                  \item Set $W_j = [\a_{i_1}, \ldots, \a_{i_r}]$ for the columns $\a_{i_1}, \ldots, \a_{i_r}$ of $A$.
              \end{enumerate}
        \item If $\tau = 1$, set $\hat{W} = W_1$ and go to Step~7; otherwise go to Step~5.
        \item Initialize $C \gets W_1$. For $j = 2, \ldots, \tau$, do:
              \begin{enumerate}[label={5-\arabic*.}]
                  \item Update $C \gets \tfrac{1}{j-1} ( W_1 + \cdots + W_{j-1} )$.
                  \item Rearrange the columns of $W_j$ so as to minimize the sum of MRSA values
                        between corresponding columns of $C$ and $W_j$.
              \end{enumerate}
        \item Set $\hat{W} = \tfrac{1}{\tau} ( W_1 + \cdots + W_{\tau} )$.
        \item Return the columns $\hat{\w}_1, \ldots, \hat{\w}_r$ of $\hat{W}$ and terminate.
    \end{enumerate}
\end{algorithm}

\subsection{Discussion on Prior Work} \label{subsec: discussion on prior work}
We discuss related work on REDIC for endmember extraction.
Ifarraguerri and Chang \cite{Ifa99} proposed convex cone analysis, which aims to identify endmember signatures from HSIs.
Their method can be viewed as based on a geometric interpretation of the LMM.
The fundamental idea of their algorithm is
to first identify multiple extreme rays of the cone generated by the columns of the input HSI matrix
and then select a subset of these rays as endmember signatures.
Similar to REDIC, their method identifies a set of extreme rays, although for a different purpose.
Soussen et al.\ \cite{Sou09} used convex cone analysis to reduce the computational burden of endmember extraction methods.
For an HSI matrix $A$,
their method heuristically computes an index set $\KC$ such that $\cone(A) = \cone(A(\KC))$
and then applies a Bayesian positive source separation algorithm to the submatrix $A(\KC)$.
Awasthi, Kalantari, and Zhang \cite{Awa18} developed an algorithm for detecting the vertices of the convex hull of a
set of points, and applied it to the computation of near-separable NMFs.
Their method enumerates all vertices of the convex hull of the columns of the input matrix
using the Triangle Algorithm \cite{Kal15} as a subroutine, and then selects a subset of these vertices as the output.
Hence, it shares similarities with the work of Ifarraguerri and Chang \cite{Ifa99}.
The authors of  \cite{Awa18} analyzed the robustness of their method to noise and its computational complexity.
As mentioned in Section~\ref{sec: introduction},
Esser et al.\ \cite{Ess12} and Gillis and Luce \cite{Gil18} used clustering-based methods
to eliminate redundant columns from HSI matrices,
motivated by the high computational cost of self-dictionary methods in endmember extraction.
Kumar, Sindhwani, and Kambadur \cite{Kum13} developed an algorithm called XRAY
for computing near-separable NMFs, which can be viewed as a greedy method.
XRAY identifies the extreme rays of the cone generated by the columns of the input matrix
by solving nonnegative least-squares problems multiple times.
The algorithm takes an estimate of the factorization rank $r$ as input and identifies extreme rays sequentially
until $r$ extreme rays have been identified.
Gillis \cite{Gil14c} developed an algorithm called SNPA by modifying XRAY.
SNPA was shown to be robust to noise both theoretically and experimentally.

\section{Experiments}\label{sec: experiments}
\subsection{Objectives, Implementation Details, and Datasets} \label{subsec: objectives, implementation details, and datasets}
We implemented REDIC in MATLAB and evaluated its performance through numerical experiments.
The first experiment (Section~\ref{subsec: performance evaluation of DRS}) was conducted
to examine the behavior of the DRS module incorporated into REDIC,
in terms of the size of the resulting matrix, its ability to retain the endmembers present in the original HSI data,
and the computation time required to perform DRS.
The second experiment (Section~\ref{subsec: performance evaluation of REDIC}) was conducted
to evaluate the performance of REDIC in terms of both endmember extraction accuracy and computation time.
All experiments were carried out in MATLAB R2024b on dual Intel Xeon Gold 6336Y processors with 256 GB of memory.

In what follows, we describe the implementation details and datasets used in our experiments.
REDIC performs DRS, which in turn invokes DR to find a set $\KC \in \Gamma(A)$ for a matrix $A$.
DR determines the feasibility of whether $\a_i \in \cone(A(\KC_\ell - i))$
by solving problem~\eqref{eq: NNLS for feasibility test}.
We used the MATLAB function \texttt{lsqnonneg}, which implements the active set method, to solve this problem.
The tolerance parameter $\epsilon_{\feas}$ in the feasibility test was set to $10^{-8}$.
REDIC also performs RCE to solve the model $\HT$.
For the implementation of RCE, we used a reorganized and speed-improved version of the MATLAB code developed in \cite{Miz25},
in which CPLEX is employed as the LP solver.
The comparison of computation time with the previous version of the RCE code is described in Remark~\ref{rem: comparison of computation time with previous RCE code}.
RCE first solves a small-scale subproblem of $\HT$ and then gradually increases the problem size until the optimal solution of $\HT$ is obtained.
It includes the parameters $\zeta$ and $\eta$ to control the size of the initial subproblem.
These parameter values were set as follows to maintain consistency with the experimental setup in \cite{Miz25}:
for an input matrix with $n$ columns,
if $n \le 300$, then $(\zeta, \eta) = (0, n)$; if $300 < n \le 50000$, then $(\zeta, \eta) = (10, 100)$;
and if $n > 50000$, then $(\zeta, \eta) = (50, 300)$.

We employed three real HSI datasets in our experiments: Jasper Ridge, Samson, and Urban, which are widely used for evaluating endmember extraction methods.
These datasets were constructed following the procedure described in \cite{Zhu17} and are the same as those used in \cite{Miz25}.
The details of each dataset are as follows:
\begin{itemize}
    \item \textbf{Jasper Ridge:} The Jasper Ridge HSI was acquired by the AVIRIS sensor and consists of $512 \times 614$ pixels with 224 bands
          covering wavelengths from 360 nm to 2500 nm.
          We used a subimage of size $100 \times 100$ pixels extracted from the original image, containing 198 valid bands
          after removing the degraded ones.
          This region includes four endmembers: Tree, Water, Soil, and Road.

    \item \textbf{Samson:} The Samson HSI consists of $952 \times 952$ pixels with 156 bands covering wavelengths from 400 nm to 900 nm.
          We used a subimage of size $95 \times 95$ pixels from the original image with all 156 bands retained.
          This region includes three endmembers: Soil, Tree, and Water.

    \item \textbf{Urban:} The Urban HSI was acquired by the HYDICE sensor and consists of $307 \times 307$ pixels with 210 bands
          covering wavelengths from 400 nm to 2500 nm.
          We used the entire image with 162 valid bands after removing the degraded ones.
          This region contains four to six endmembers.
          Following the setting in \cite{Gil18, Miz25}, we set the number of endmembers to six: Asphalt, Grass, Tree, Roof~1, Roof~2, and Soil.
\end{itemize}

Table~\ref{tab: dataset specifications} summarizes, for each dataset,
the size of the HSI matrix $A \in \Real^{d \times n}$,
the number of endmembers $r$, and the specific names of the endmembers considered in our experiments.
Figure~\ref{fig: rgb images of datasets} shows the RGB images of these datasets.
\begin{table}[h]
    \centering
    \footnotesize
    \caption{Specifications of the HSI datasets.}
    \label{tab: dataset specifications}
    \begin{tabular}{l l l}
        \toprule
        Dataset      & $(d, n, r)$       & Endmembers                                 \\
        \midrule
        Jasper Ridge & $(198, 10000, 4)$ & Tree, Water, Soil, Road                    \\
        Samson       & $(156, 9025, 3)$  & Soil, Tree, Water                          \\
        Urban        & $(162, 94249, 6)$ & Asphalt, Grass, Tree, Roof~1, Roof~2, Soil \\
        \bottomrule
    \end{tabular}
\end{table}

\begin{figure}[h]
    \centering
    \includegraphics[scale=0.38]{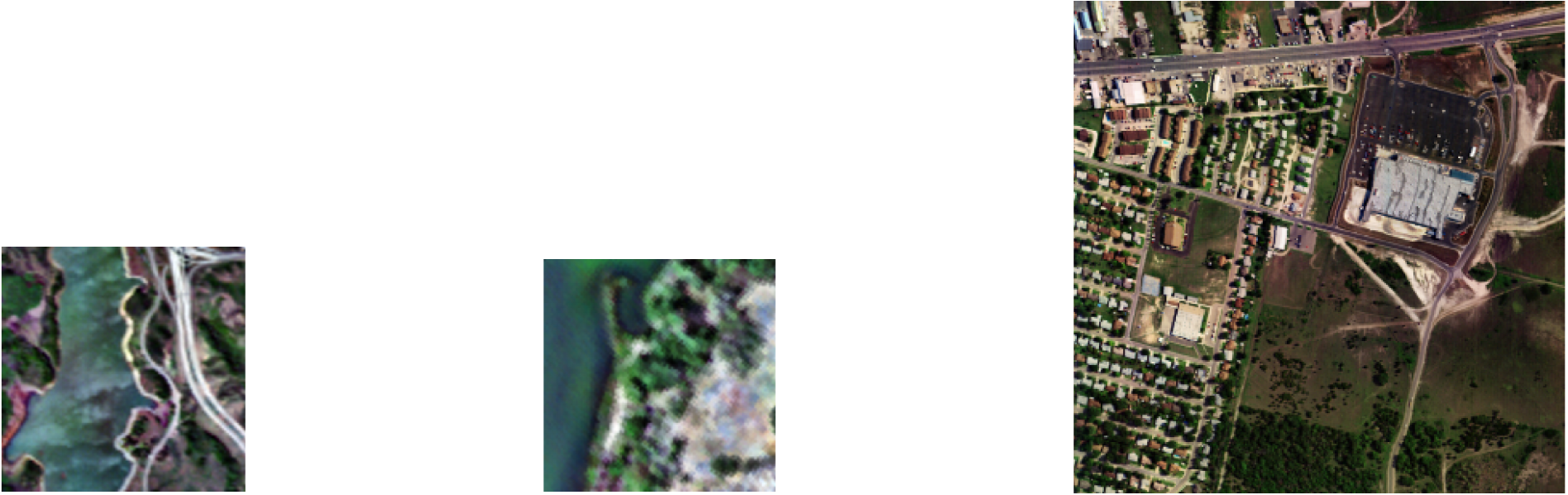}
    \caption{RGB images of the Jasper Ridge (left), Samson (center), and Urban (right) datasets.}
    \label{fig: rgb images of datasets}
\end{figure}

\subsection{Performance Evaluation of DRS} \label{subsec: performance evaluation of DRS}
\subsubsection{Experimental Setup} \label{subsubsec: experimental setup for DRS experiments}
We present experimental results evaluating the output of DRS and its computation time
in Section~\ref{subsubsec: evaluation of the DRS output} and Section~\ref{subsubsec: computational time of DRS}, respectively.
Here, we describe the experimental setup for these experiments.

\paragraph{Tested Algorithms}
We evaluated the DRS module incorporated into REDIC (Steps~1 and 2, and part of Step~3).
For clarity, we briefly describe the tested algorithms below.
The algorithm takes a matrix $A$ and parameters $r$, $p$, and $\lambda$ as input:
\begin{enumerate}[label={\arabic*.}]
    \item Construct the reduced matrix $A'$ by computing the top-$r$ truncated SVD of $A$.
    \item Apply DRS to $A'$ with parameter $p$ to obtain $\KC \in \Gamma(A')$.
    \item Repeat this procedure 50 times to generate 50 sets of $\KC_{\add}$.
          If $\lambda - |\KC| \le 0$, then all sets are empty; otherwise, each set is constructed
          by randomly selecting $\lambda - |\KC|$ elements from $[n] \setminus \KC$.
\end{enumerate}
For comparison, we also evaluated the performance of DR.
The tested DR algorithm corresponds to Steps~1 and~2 of the above procedure, with DRS replaced by DR:
it takes a matrix $A$ and a parameter $r$ as input,
computes the reduced matrix $A'$ through the top-$r$ truncated SVD of $A$,
and then applies DR to $A'$ to obtain $\KC \in \Gamma(A')$.

\paragraph{Datasets Used for Evaluation}
We used synthetic HSI datasets to evaluate the output of DRS and the three real HSI datasets to evaluate its computation time.
The synthetic datasets were generated from the three real datasets following the procedure described in Section~VII-B of \cite{Miz25}.
To make this paper self-contained, we summarize the procedure here.
The endmember signatures for the real HSI datasets were identified in \cite{Zhu17}, and
these signatures were used to generate the synthetic datasets.
Let $A^{\real} \in \Real^{d \times n}$ denote the HSI matrix for a given HSI dataset,
and let $\w_1^{\ident}, \ldots, \w_r^{\ident}$ represent the identified endmember signatures.
The following procedure was used to construct $W \in \Real^{d \times r}$, $H \in \Real^{r \times n}$, and $V \in \Real^{d \times n}$:
\begin{enumerate}[label={\arabic*.}]
    \item Normalize all columns of $A^{\real}$ to have unit $L_1$ norm.
    \item For $j = 1, \ldots, r$, compute
          \[
              i_j = \arg\min_{i \in [n]} \MRSA(\w_j^{\ident}, \a^{\real}_i),
          \]
          where $\a^{\real}_i$ is the $i$th column of $A^{\real}$.
    \item Set $W = A^{\real}(\{i_1, \ldots, i_r\})$.
    \item Compute the optimal solution $X$ to the following convex optimization problem:
          \[
              \min_{X \in \Real^{r \times n}} \ \|A^{\real} - WX\|_F^2
              \quad \text{subject to} \quad \one^\trans X = \one^\trans,\ X \ge O,
          \]
          and set $X(\{i_1, \ldots, i_r\}) = I$.
    \item Set $H = X$ and $V = A^{\real} - WH$.
\end{enumerate}

Three synthetic datasets were generated using the matrices $W$, $H$, and $V$ obtained through the above procedure:
Dataset~1 from Jasper Ridge, Dataset~2 from Samson, and Dataset~3 from Urban.
Each dataset consists of HSI matrices of the form
$ A = WH + (\nu / \| V \|_1) \cdot V \in \Real^{d \times n}, $
where the noise intensity level $\nu$ was varied from $0$ to $1.5$ in increments of $0.1$,
resulting in 16 HSI matrices with different noise levels.
The matrix size $(d,n)$ and the number of endmembers $r$ for Datasets~1--3 are identical to those of the corresponding original HSI datasets.
When $\nu = \|V\|_1$,
the HSI matrix $A$ coincides with $A^{\real}$, whose columns are all normalized to have unit $L_1$ norm.
In particular, $\|V\|_1$ equals 0.61, 0.15, and 1.15 for Datasets~1--3, respectively.

\paragraph{Evaluation Metrics}
We used the following metrics to evaluate the output of DRS.
The reconstruction error was introduced to assess
how well the original reduced matrix $A'$ can be reconstructed from the submatrix $A'(\KC)$.
This metric is defined as follows. Recall that $A'$ is of size $r \times n$, with columns $\a'_1, \ldots, \a'_n$.
First, we compute the optimal solutions $\x_1, \ldots, \x_n$ to the following nonnegative least-squares problems:
\[
    \x_i = \arg\min_{\x \ge \zero} \, \| A'(\KC)\x - \a'_i \|_2^2 \quad \text{for } i \in [n].
\]
The reconstruction error of $A'$ based on $A'(\KC)$ is then defined as
\[
    \sqrt{\frac{1}{rn} \sum_{i=1}^{n} \| A'(\KC)\x_i - \a'_i \|_2^2}.
\]

The distance between $A(\KC \cup \KC_{\add})$ and $W$ was introduced
to quantify how closely each column of $W$ can be matched by at least one column of $A(\KC \cup \KC_{\add})$.
Let $\dist(\a, \b)$ denote either the $L_1$ distance or the MRSA value between $\a, \b \in \Real^d$.
For each column $\w_i$ of $W$ with $i \in [r]$,
we select an index $k_i$ from $\KC \cup \KC_{\add}$ corresponding to the column of $A$ closest to $\w_i$ under the chosen metric:
\[
    k_i = \arg\min_{k \in \KC \cup \KC_{\add}} \, \dist(\w_i, \a_k) \quad \text{for } i \in [r].
\]
We then define the distance between $A(\KC \cup \KC_{\add})$ and $W$
as the mean of these distances,
\[
    \frac{1}{r} \sum_{i=1}^r \dist(\w_i, \a_{k_i}),
\]
which we refer to as the \emph{$L_1$ distance (resp.\ MRSA) of the DRS output $\KC \cup \KC_{\add}$}
with respect to the reference endmember matrix $W$, when $\dist$ represents the $L_1$ distance (resp.\ MRSA value).
These metrics were used to assess the similarity between $A(\KC \cup \KC_{\add})$ and $W$.

\subsubsection{Evaluation of the DRS Output} \label{subsubsec: evaluation of the DRS output}
We conducted experiments to evaluate the output of DRS.
We ran the DRS algorithm, described in Section~\ref{subsubsec: experimental setup for DRS experiments}, on Datasets~1--3.
The algorithm uses parameters $r$, $p$, and $\lambda$.
The parameter $\lambda$ represents the number of elements added in the data augmentation technique,
while $p$ represents the number of clusters used in the data-splitting technique.
We varied $\lambda$ from 0 to 250 in increments of 50 for Datasets~1 and~2,
and from 0 to 2500 in increments of 500 for Dataset~3.
Each value of $\lambda$ corresponds to approximately $0\%$--$2.5\%$ of the number of columns $n$ in each dataset,
with increments of 0.5\%.
In addition, for Dataset~3, $\lambda$ was further varied from 0 to 12500 in increments of 2500,
corresponding to approximately $0\%$--$12.5\%$ of the number of columns $n$, with increments of 2.5\%.
We set $p = 30$ in these experiments and set $r$ equal to the number of endmembers for each dataset.
After running the algorithm, we evaluated the following quantities:
(i) the size of $\KC$,
(ii) the reconstruction error of $A'$ based on $A'(\KC)$, and
(iii) the distance between $A(\KC \cup \KC_{\add})$ and $W$.

\begin{figure}[t]
    \centering

    \begin{subfigure}[t]{0.48\linewidth}
        \centering
        \includegraphics[width=\linewidth]{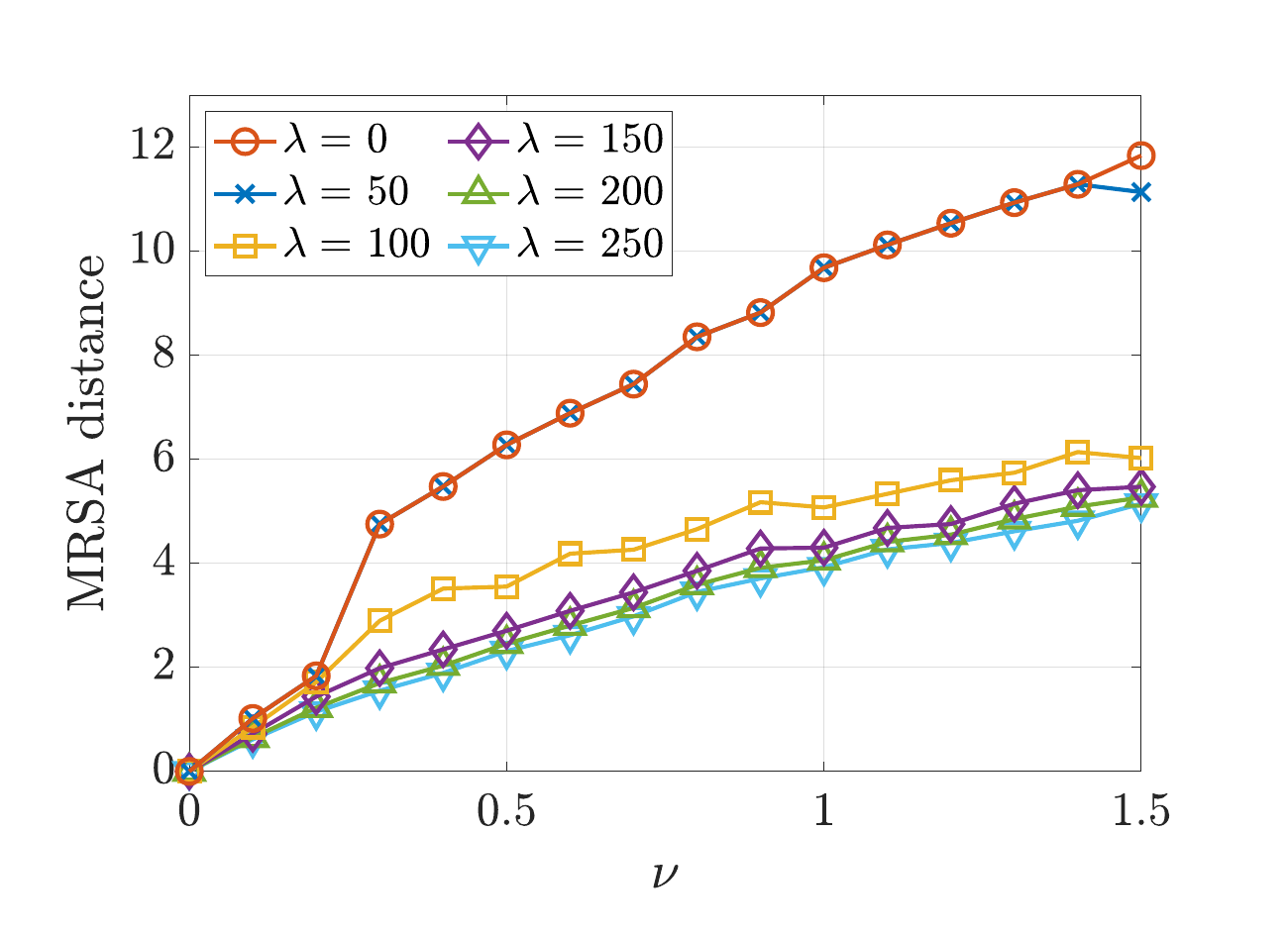}
        \caption{Dataset~1}
    \end{subfigure}
    \hfill
    \begin{subfigure}[t]{0.48\linewidth}
        \centering
        \includegraphics[width=\linewidth]{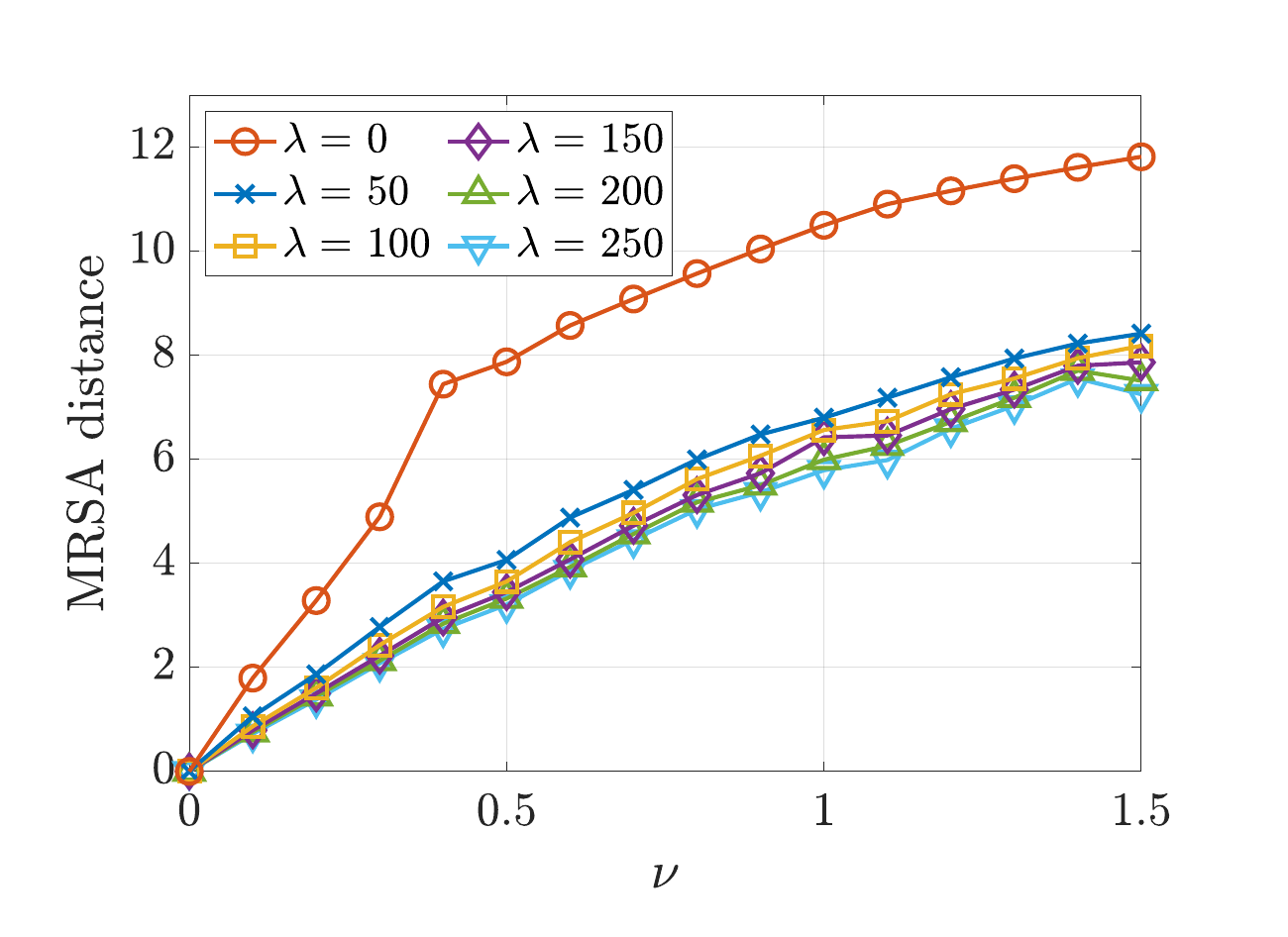}
        \caption{Dataset~2}
    \end{subfigure}

    \vspace{3mm}
    \begin{subfigure}[t]{0.48\linewidth}
        \centering
        \includegraphics[width=\linewidth]{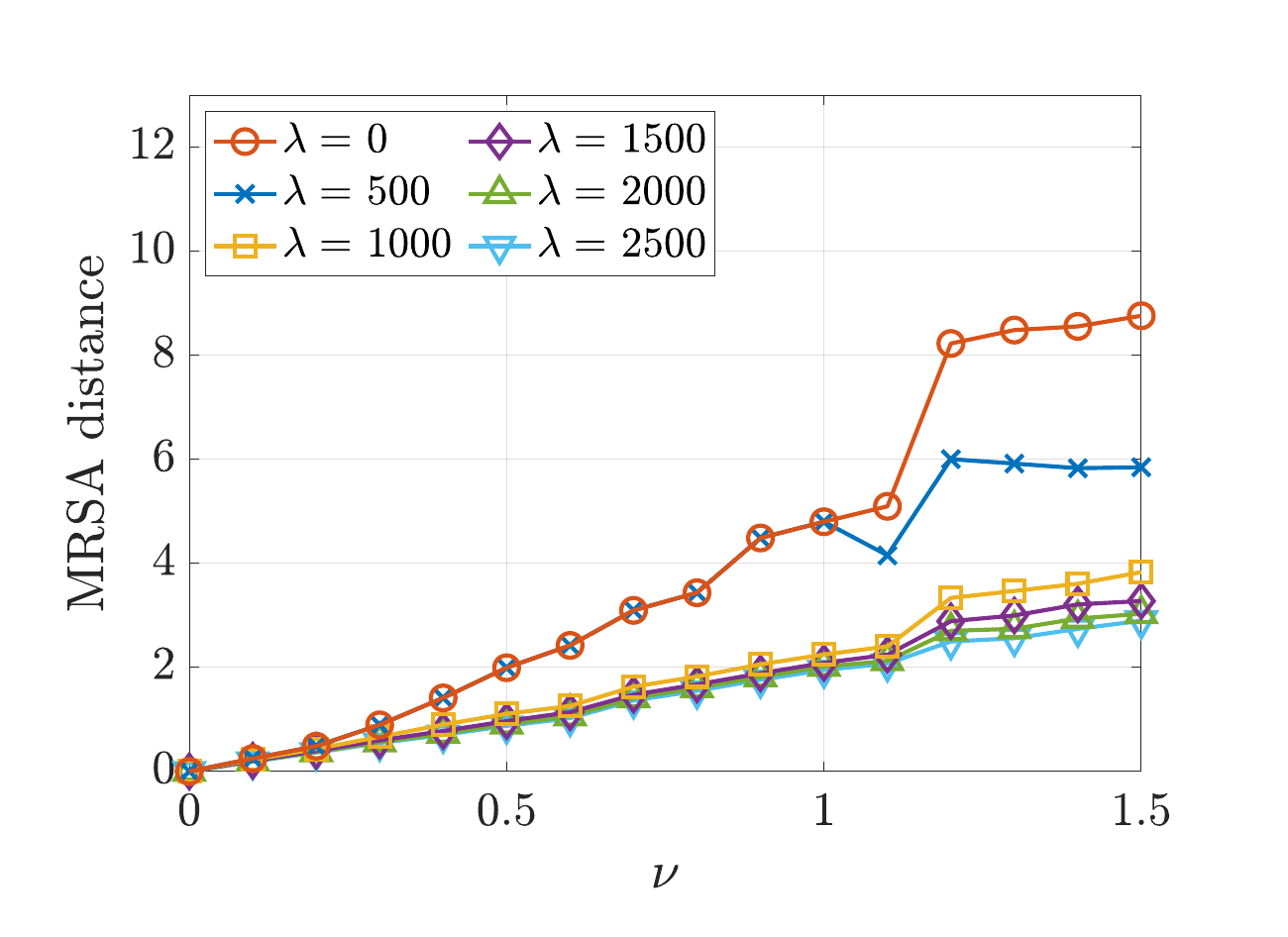}
        \caption{Dataset~3 with $\lambda$ from 0 to 2500 (step 500)}
    \end{subfigure}
    \hfill
    \begin{subfigure}[t]{0.48\linewidth}
        \centering
        \includegraphics[width=\linewidth]{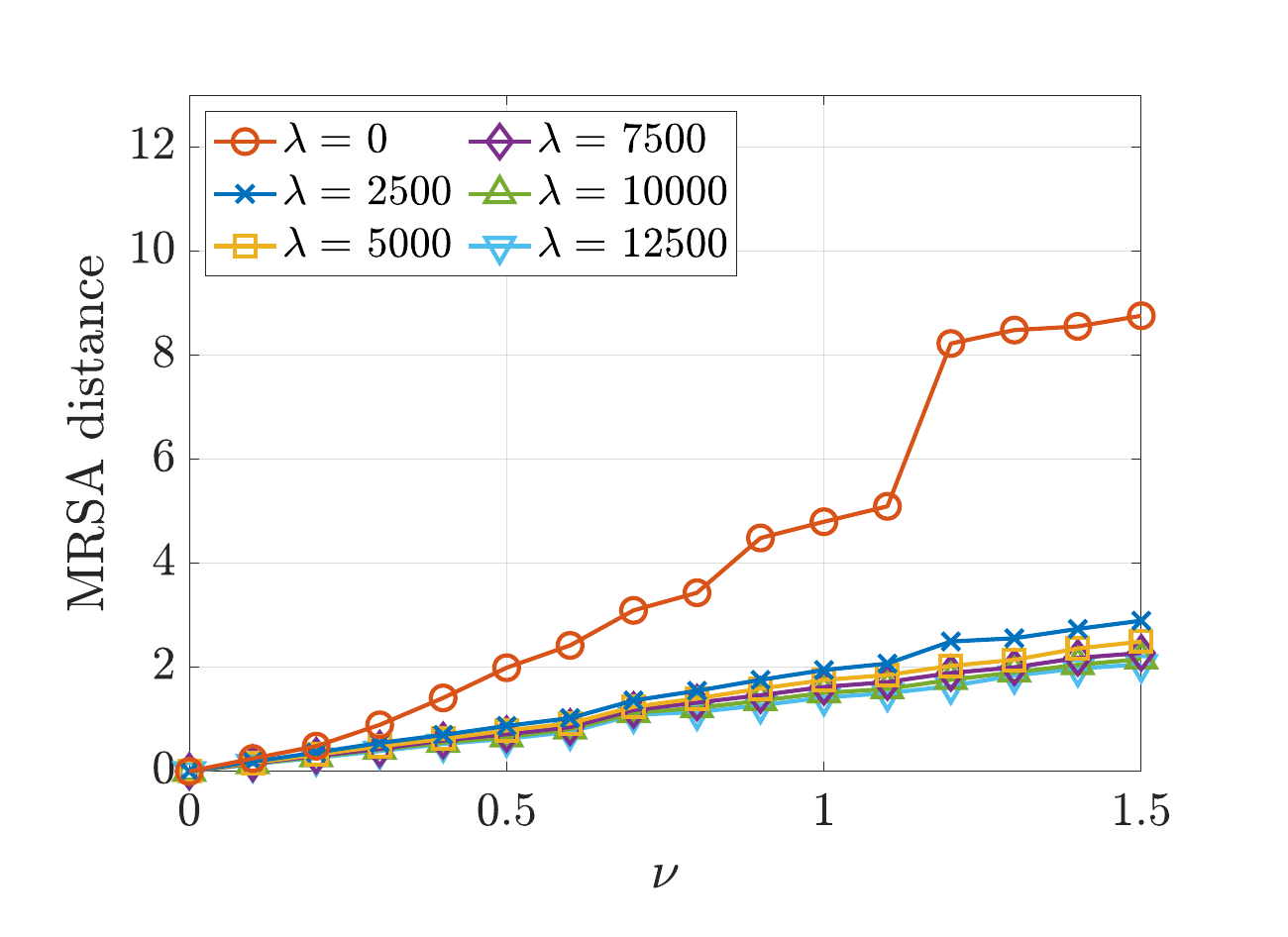}
        \caption{Dataset~3 with $\lambda$ from 0 to 12500 (step 2500)}
    \end{subfigure}

    \caption{Experimental results for Datasets~1--3 using the MRSA distance.
        Panels (a) and (b) show Datasets~1 and~2, while
        Panels (c) and (d) show Dataset~3 for different ranges of $\lambda$.
        Each panel shows the MRSA distance of the DRS output $\KC \cup \KC_{\add}$,
        averaged over 50 realizations of $\KC_{\add}$.}
    \label{fig: results of DRS for datasets 1-3 using the MRSA distance}
\end{figure}

We first examine the size of the output set $\KC$ produced by DRS.
For each dataset, this size matches the number of endmembers when $\nu = 0$.
When $\nu > 0$, the average sizes of $\KC$ are 64.6 for Dataset~1, 16.4 for Dataset~2, and 546.8 for Dataset~3.
In other words, DRS retains only about 0.2\%--0.6\% of the columns for each dataset.
Next, we evaluate the reconstruction error of $A'$ based on $A'(\KC)$.
The error values are below $10^{-8}$ for all datasets and all noise intensity levels $\nu$.
We then evaluate the distance between $A(\KC \cup \KC_{\add})$ and $W$.
Since the experimental results based on the $L_1$ distance and the MRSA distance showed similar trends,
we report only the results based on the MRSA distance;
the results based on the $L_1$ distance are provided in Appendix~\ref{sec: supplementary experimental results}.
Figure~\ref{fig: results of DRS for datasets 1-3 using the MRSA distance} summarizes the experimental results for Datasets~1--3,
showing the MRSA distance of the DRS output $\KC \cup \KC_{\add}$ at each noise intensity level~$\nu$.
Each plot shows the mean distance over 50 realizations of $\KC_{\add}$ for each value of $\lambda$.
The figures show that the distances decrease as $\lambda$ increases for all datasets.
In particular, even when $\lambda$ is set to 1\% of $n$, the distances are substantially reduced compared with those at $\lambda = 0$.

\subsubsection{Computational Time of DRS} \label{subsubsec: computational time of DRS}

\begin{figure}[t]
    \centering
    \begin{subfigure}[t]{0.48\linewidth}
        \centering
        \includegraphics[width=\linewidth]{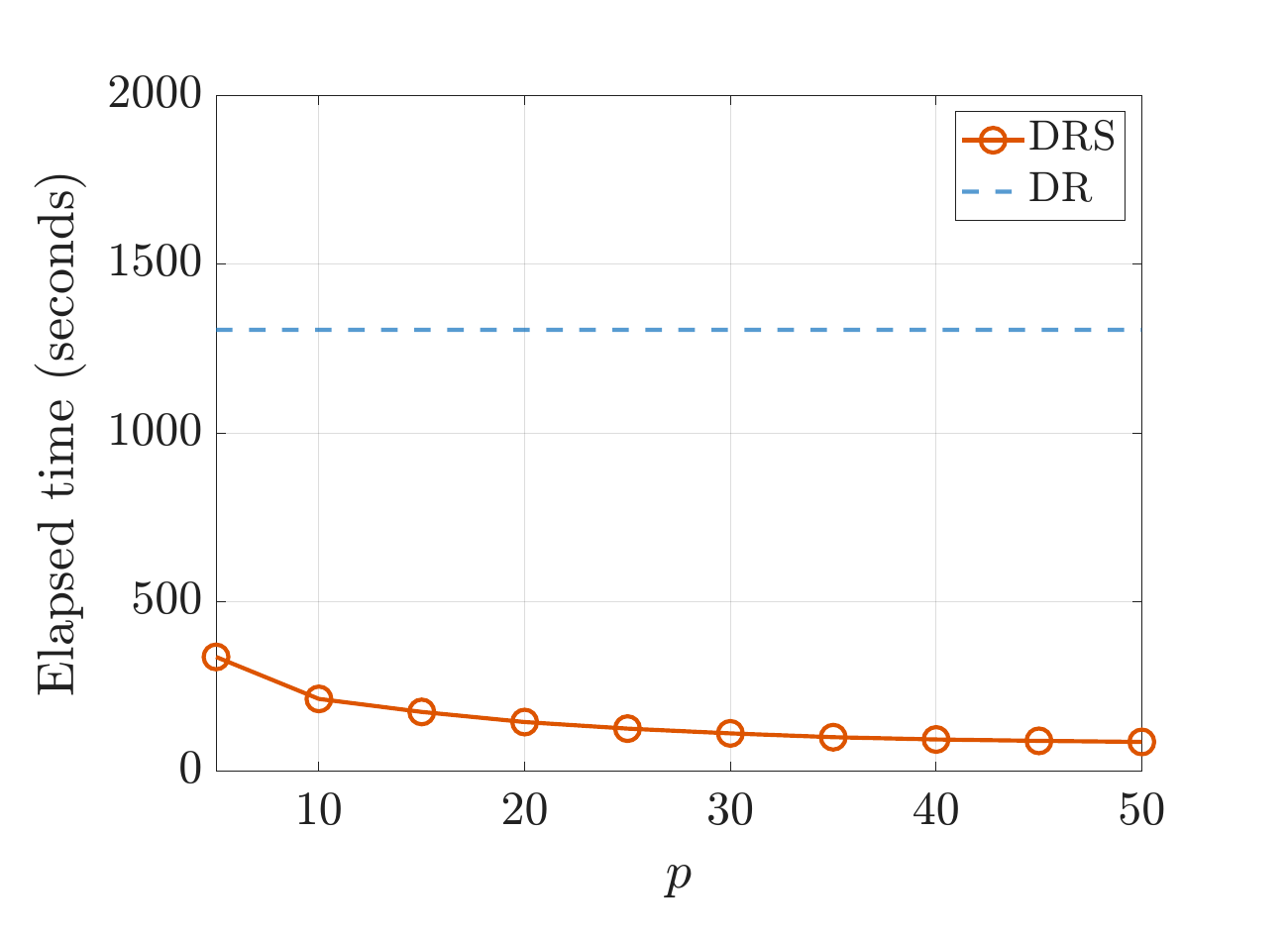}
        \caption{Elapsed time (seconds)}
    \end{subfigure}
    \hfill
    \begin{subfigure}[t]{0.48\linewidth}
        \centering
        \includegraphics[width=\linewidth]{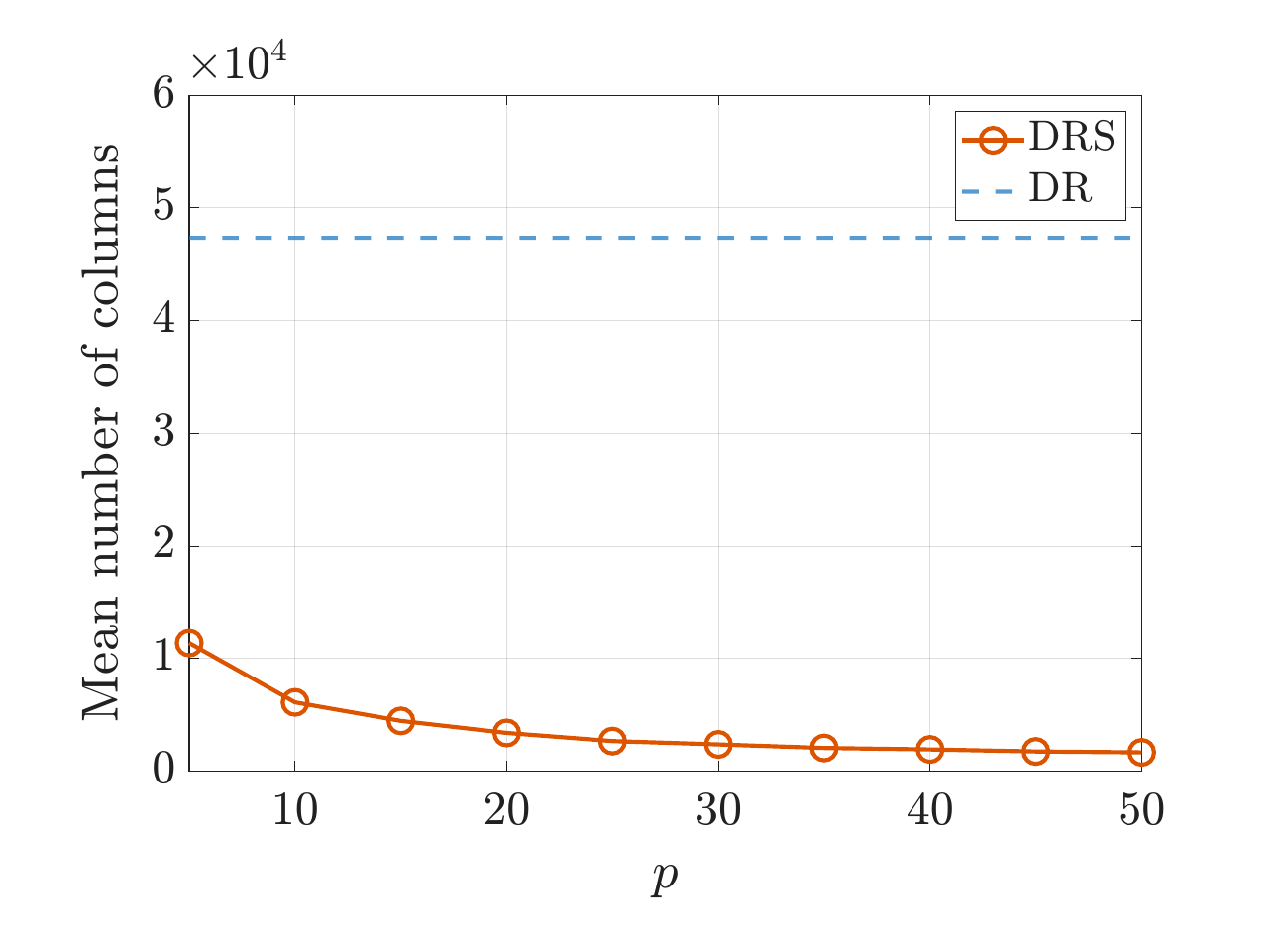}
        \caption{Mean number of columns}
    \end{subfigure}

    \vspace{3mm}
    \begin{subfigure}[t]{0.48\linewidth}
        \centering
        \includegraphics[width=\linewidth]{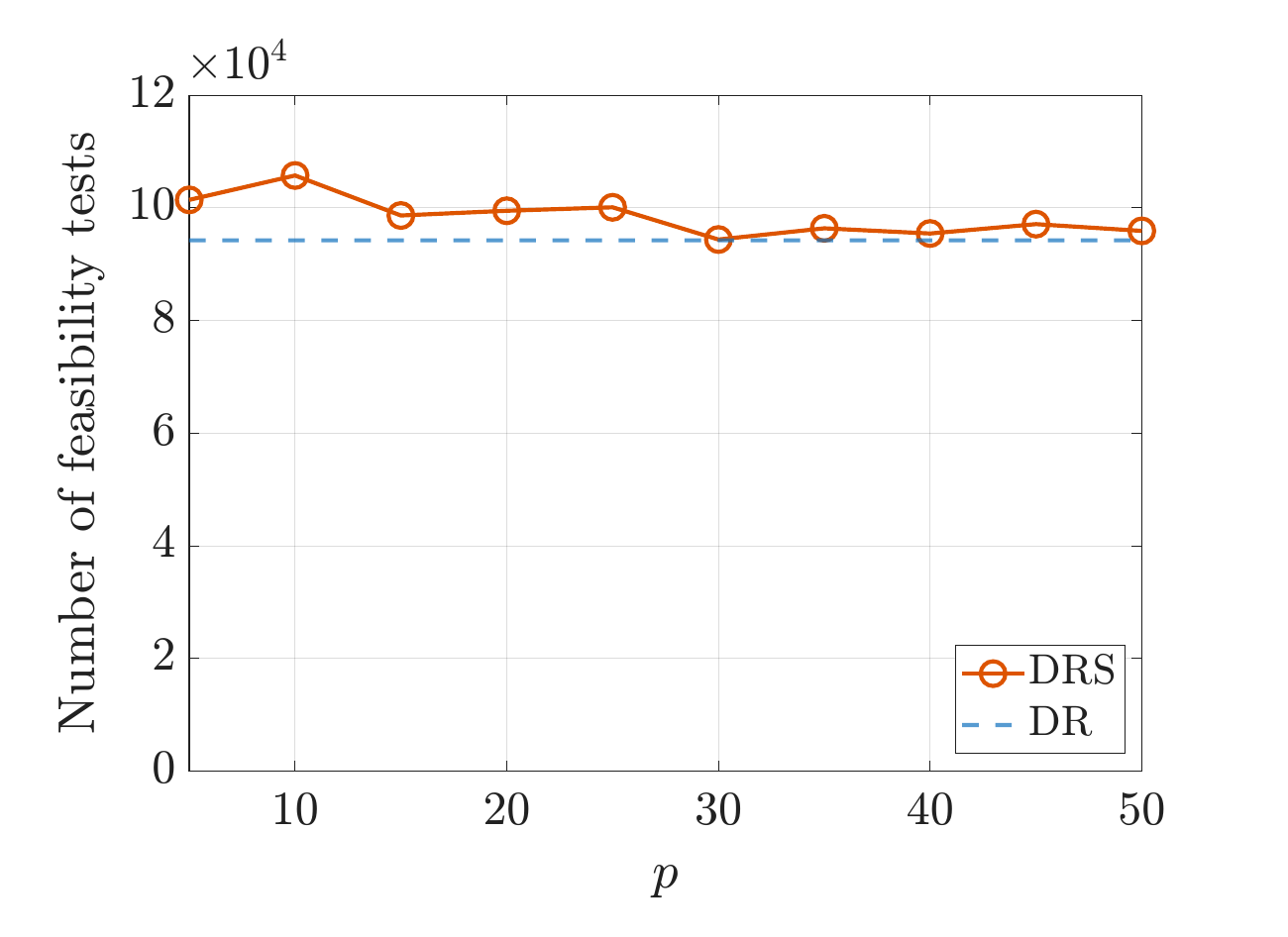}
        \caption{Number of feasibility tests}
    \end{subfigure}
    \caption{Results of DRS and DR on the Urban dataset.
        Panel (a) shows the elapsed time (seconds),
        Panel (b) the mean number of columns processed in the feasibility tests (Step~2 of DR),
        and Panel (c) the number of feasibility tests.}
    \label{fig: results of DRS for the Urban data}
\end{figure}

We conducted experiments to evaluate the computational time of DRS.
We ran the DRS and DR algorithms, described in Section~\ref{subsubsec: experimental setup for DRS experiments},
on the three real HSI datasets.
In these experiments, the parameters $p$ and $\lambda$ in DRS were set as follows:
$\lambda = 0$, while $p$ was varied from 5 to 50 in increments of 5.
The parameter $r$ for both DRS and DR was set equal to the number of endmembers in each dataset.
Here, we report the results on the Urban dataset in detail, as the results on the three datasets exhibited similar trends.

Figure~\ref{fig: results of DRS for the Urban data} summarizes the experimental results for the Urban dataset:
the elapsed time (Panel (a)), the mean number of columns in the matrices processed in Step~2 of DR during the iterations (Panel (b)),
and the number of feasibility tests (Panel (c)) for DRS at each value of $p$ and for DR.
We provide a more detailed explanation of Panel (b).
For each iteration,
we counted the number of columns in the matrix $A(\KC_\ell -i)$ processed in Step~2 of DR.
Since DRS calls DR multiple times, we computed the total number of columns in such matrices over all iterations of DR invoked by DRS.
Panel (b) shows the mean of these numbers over all iterations for DRS and DR.

Panel (a) shows that DRS is faster than DR.
Specifically, DRS takes less than 6 minutes for all values of $p$, and in particular, it takes 2 minutes when $p = 30$,
whereas DR takes 22 minutes.
The elapsed time of DRS gradually decreases as $p$ increases, and the rate of decrease becomes smaller for $p \ge 30$.
Panels (b) and (c) show that
DRS substantially reduces the size of problem~\eqref{eq: NNLS for feasibility test} in the feasibility tests,
while the number of feasibility tests performed by DRS is nearly the same as that of DR.
These results indicate that, with an appropriate choice of $p$,
DRS reduces the problem sizes arising in the feasibility tests of DR while maintaining a similar number of feasibility tests to DR,
thereby reducing the elapsed time.

Finally, we report the elapsed time of DRS and DR for the Jasper Ridge and Samson datasets.
For the Jasper Ridge dataset, DR takes 21 seconds, whereas DRS with $p = 30$ takes 5 seconds.
For the Samson dataset, DR takes 8 seconds, whereas DRS with $p = 30$ takes 2 seconds.

\subsection{Performance Evaluation of REDIC} \label{subsec: performance evaluation of REDIC}
\subsubsection{Experimental Setup} \label{subsubsec: experimental setup for REDIC experiments}
The experiments evaluating REDIC are presented in Section~\ref{subsubsec: comparison with existing methods},
and the ablation studies are presented in Section~\ref{subsubsec: ablation studies}.
Here, we describe the experimental setup for these experiments.

\paragraph{Tested Algorithms}
Section~\ref{subsubsec: comparison with existing methods} presents a comparison between REDIC and existing endmember extraction methods.
The methods included in the comparison are SPA \cite{Gil14a}, ER \cite{Miz14}, VCA \cite{Nas05}, SNPA \cite{Gil14c}, and MERIT \cite{Ngu22},
all of which were compared with the LP method in \cite{Miz25}.
Here, ER is a variant of SPA with a preprocessing step based on the minimum-volume ellipsoid of the columns of the input matrix.
These methods can be categorized as follows: MERIT is a convex programming--based method, whereas the other methods are greedy methods.
Note that we also conducted experiments using PSPA \cite{Gil15, Miz16}, which is a preconditioned variant of SPA.
However, the results are not included in this paper because the performance metrics varied across MATLAB versions;
see Remark~\ref{rem: performance of PSPA} for details.
In addition, we included MinVol \cite{Gil20, Lep20} in the comparison.

The experimental setup for these methods, except for MinVol, was the same as that in \cite{Miz25}.
To avoid confusion, we briefly describe the version of the MERIT algorithm evaluated in this experiment.
Like REDIC, the tested version of MERIT incorporates an SVD-based dimensionality reduction technique:
it first constructs a reduced matrix $A'$ by computing the top-$r$ truncated SVD of the input matrix $A$,
and then applies MERIT to $A'$ to extract endmember signatures.

We briefly describe the MinVol method below.
The method employs an NMF formulation with a minimum-volume regularization term.
Although various variants of MinVol have been proposed,
we adopted the version referred to as min-vol NMF (3) in Section~4.3.3 of \cite{Gil20} for our experiments,
because the textbook provides numerical evidence demonstrating that this version is effective for endmember extraction from the Urban dataset.
The method solves the following optimization problem:
\begin{align} \label{eq: min-vol NMF formulation}
    \min_{W \in \WC, \ H \in \HC} \  & \|A - WH\|_F^2  + \lambda \log\det(W^\trans W + \delta I)
\end{align}
where $\WC = \{ W \in \Real^{d \times r} \mid W^\trans \one = \one, \ W \ge O \}$ and $\HC = \{ H \in \Real^{r \times n} \mid H \ge O \}$.
The parameters $\lambda$ and $\delta$ control the strength of the regularization term and the numerical stability, respectively.
For the implementation of MinVol,
we used the MATLAB function \texttt{minvolNMF}, which is available in the GitLab repository maintained by the author of \cite{Gil20}.

\paragraph{Evaluation Metric}
We employed the following metric to assess the accuracy of endmember extraction.
Let $A = [\a_1, \ldots, \a_n] \in \Real^{d \times n}$ denote the HSI matrix for a given dataset,
and let $\w_1^{\ident}, \ldots, \w_r^{\ident} \in \Real^{d}$ represent the endmember signatures identified in \cite{Zhu17}.
We compute the column indices $i_1, \ldots, i_r$ of $A$ corresponding to $\w_1^{\ident}, \ldots, \w_r^{\ident}$ under the MRSA metric as
\begin{align*}
    i_j = \arg\min_{i \in [n]} \MRSA(\w_j^{\ident}, \a_i) \quad \text{for } j \in [r],
\end{align*}
and regard $\a_{i_1}, \ldots, \a_{i_r}$ as the reference endmember signatures.
This procedure was introduced to ensure the reproducibility of the experimental results.
Let $\w_1, \ldots, \w_r$ denote the reference endmember signatures,
and let $\hat{\w}_1, \ldots, \hat{\w}_r$ denote the estimated endmember signatures obtained by a given method.
We then determine a permutation $\sigma$ on the set $[r]$ according to
\begin{align*}
    \sigma = \arg\min_{\sigma \in \SC_r} \sum_{j=1}^r \MRSA(\w_{\sigma(j)}, \hat{\w}_j),
\end{align*}
where $\SC_r$ denotes the set of all permutations on $[r]$.
Such a permutation can be obtained by solving an assignment problem.
Finally, we compute the MRSA values between the reference and estimated endmembers,
$\MRSA(\w_{\sigma(j)}, \hat{\w}_j)$ for $j \in [r]$,
and their mean value,
\begin{align*}
    \frac{1}{r} \sum_{j=1}^r \MRSA(\w_{\sigma(j)}, \hat{\w}_j).
\end{align*}
We define this mean value as the \emph{MRSA score} of the method.

\paragraph{Data-Specific Preprocessing for the Urban Dataset}
This experiment was based on the EEHT-C experiment for the Urban dataset presented in \cite{Miz25},
where a data-specific preprocessing technique was applied prior to executing EEHT-C
to enhance its performance (see Section~VII-C of \cite{Miz25} for details).
The preprocessing technique involves two parameters, $\phi$ and $\omega$.
In this experiment, prior to performing endmember extraction,
we applied the same preprocessing to the Urban dataset, using $(\phi, \omega) = (0.4, 0.1)$ as in \cite{Miz25}.
This procedure eliminated 2995 columns from the original HSI matrix, which contained 94249 columns,
resulting in a total of 91254 columns.
No such preprocessing was applied to the Jasper Ridge and Samson datasets.

\subsubsection{Comparison with Existing Methods} \label{subsubsec: comparison with existing methods}
We compared the performance of REDIC with that of existing endmember extraction methods
using the three real HSI datasets described in Section~\ref{subsec: objectives, implementation details, and datasets}.
In these experiments, we set the parameter $r$ equal to the number of endmembers for each dataset and $p$ to 30.
Before presenting the experimental results for REDIC,
we first report the specifications of the DRS output $\KC$ for each dataset in Table~\ref{tab: specifications of DRS output}.
These results were obtained by applying Steps~1 and~2 of REDIC.
We observe that DRS retains only about 0.2\%--0.5\% of the columns for each dataset,
that the reduced matrix $A'$ can be reconstructed from the submatrix $A'(\KC)$ with a negligible error,
and that the reduced matrices contain columns close to the reference endmember signatures.
It should be noted that these MRSA distances serve as lower bounds for the MRSA scores achievable by
applying endmember extraction methods to the reduced matrices.
Figure~\ref{fig: samson}, shown in Section~\ref{sec: introduction}, was generated
using the DRS output $\KC$ obtained by executing Steps~1 and~2 of REDIC for the Samson dataset.
\begin{table}[h]
    \centering
    \footnotesize
    \caption{Specifications of the DRS output $\KC$ for each dataset.}
    \label{tab: specifications of DRS output}
    \begin{tabular}{r c c c}
        \toprule
                             & Jasper Ridge           & Samson                 & Urban                  \\
        \midrule
        Number of elements   & 53                     & 20                     & 483                    \\
        MRSA distance        & 5.96                   & 2.48                   & 6.15                   \\
        Reconstruction error & $2.11 \times 10^{-12}$ & $2.52 \times 10^{-13}$ & $1.05 \times 10^{-16}$ \\
        \bottomrule
    \end{tabular}
\end{table}

We then present the experimental results comparing REDIC with the LP method.
These methods were applied to each HSI dataset, and their MRSA scores and computational times were compared.
The LP method considered here is identical to EEHT-C proposed in \cite{Miz25},
whose procedure was reviewed in Section~\ref{sec: REDIC method}.
The parameters $\lambda$ and $\tau$ in REDIC were set as follows.
$\tau$ took values in $\{1, 5\}$, and $\lambda$ was varied in the same manner as
in the first experiment described in Section~\ref{subsubsec: evaluation of the DRS output}.
When $\lambda > 0$, REDIC was executed 50 times for each of the three HSI datasets,
and the mean and standard deviation of the MRSA scores and elapsed times over the 50 runs were computed.
For the Urban dataset,
we report the results for $\lambda = 0$ and $\lambda$ from 2500 to 12500,
while omitting the results for $\lambda$ ranging from 500 to 2000 because they are not essential to the discussion of the experimental results.
Table~\ref{tab: experimental results of REDIC on the three datasets} summarizes the experimental results.
We note that the standard deviation of the MRSA scores is zero for Jasper Ridge when $\lambda = 50$.
This is because, as shown in Table~\ref{tab: specifications of DRS output},
the DRS output $\KC$ for Jasper Ridge contains 53 elements, and therefore $\KC_{\add} = \emptyset$ when $\lambda = 50$.
We make the following observations from these results:
\begin{itemize}
    \item For all datasets,
          the MRSA scores of REDIC with $\lambda = 0$ are larger than those of the LP method,
          and noticeable gaps exist between these scores and the MRSA distances of the DRS output reported in Table~\ref{tab: specifications of DRS output}.
    \item The MRSA scores of REDIC tend to decrease as $\lambda$ increases for all datasets.
    \item When $\lambda > 0$, the mean and standard deviation of the MRSA scores with $\tau = 5$ are lower than those with $\tau = 1$.
    \item For Jasper Ridge and Samson, REDIC with $(\lambda, \tau) = (100, 5)$ achieves nearly the same MRSA scores as the LP method
          while reducing the computation time, yielding speed-ups of 33 and 5 times, respectively.
          In other words, adding only about 1\% of the columns was sufficient for extracting endmembers with accuracy comparable to that of the LP method.
          Moreover, the MRSA scores decrease further as $\lambda$ increases.
    \item For Urban, REDIC with $(\lambda, \tau) = (10000, 5)$ achieves nearly the same MRSA scores as the LP method, with a speed-up of 29 times.
          In contrast to Jasper Ridge and Samson, about 10\% of the columns had to be added to achieve comparable accuracy.
\end{itemize}

The LP method took 8 hours to process the Urban dataset, whereas REDIC completed the same task in only 16 minutes
without compromising endmember extraction accuracy.
These results suggest that incorporating DRS into self-dictionary methods
enables the efficient processing of large-scale datasets within a reasonable time frame.

\begin{remark} \label{rem: comparison of computation time with previous RCE code}
    The previous version of the MATLAB code implementing the LP method (i.e., EEHT-C) developed in \cite{Miz25} required 10 hours to process the Urban dataset.
    In contrast, the reorganized and speed-improved version of the code used in this experiment completed the same task in 8 hours.
\end{remark}
\begin{remark}
    Similar to the LP method,
    REDIC requires a data-specific preprocessing technique to achieve high endmember extraction accuracy on the Urban dataset.
    Without this preprocessing, the mean MRSA score of REDIC with $(\lambda, \tau) = (10000, 5)$ over 50 runs was 12.10.
\end{remark}

\begin{table}[t]
    \centering
    \footnotesize
    \caption{
        Results of REDIC on three real HSI datasets for different values of $\lambda$ and $\tau$.
        MRSA and Time denote the MRSA score and the elapsed time in seconds, respectively.
        The baseline corresponds to the LP method.
        For REDIC with $\lambda > 0$, the mean and standard deviation over 50 runs are shown.
    }
    \label{tab: experimental results of REDIC on the three datasets}
    \begin{subtable}[t]{1.0\linewidth}
        \caption{Jasper Ridge}
        \begin{minipage}[t]{0.26\linewidth}
            \vspace{0pt}
            \centering
            \begin{tabular}{ll}
                \toprule
                \multicolumn{2}{l}{\textbf{Baseline (LP method)}} \\
                \midrule
                MRSA & 6.82                                       \\
                Time & 271.9                                      \\
                \bottomrule
            \end{tabular}
        \end{minipage}
        \hfill
        \begin{minipage}[t]{0.71\linewidth}
            \vspace{0pt}
            \centering
            \begin{tabular}{
                    l
                    S[table-format=2.2(3)]
                    S[table-format=2.2(3)]
                    S[table-format=2.1(2)]
                    S[table-format=2.1(2)]
                }
                \toprule
                {$\lambda$}
                    & {MRSA ($\tau=1$)}
                    & {MRSA ($\tau=5$)}
                    & {Time ($\tau=1$)}
                    & {Time ($\tau=5$)}                                                               \\
                \midrule
                0   & 12.24             & \multicolumn{1}{c}{--} & 6.3       & \multicolumn{1}{c}{--} \\
                50  & 12.24             & 12.24                  & 4.9(0.4)  & 5.3(0.3)               \\
                100 & 8.49(1.10)        & 7.10(0.87)             & 5.5(0.4)  & 8.2(0.4)               \\
                150 & 7.39(1.42)        & 6.11(0.63)             & 7.5(0.4)  & 18.8(0.7)              \\
                200 & 7.04(1.22)        & 5.77(0.48)             & 10.2(0.5) & 32.7(1.0)              \\
                250 & 6.88(1.25)        & 5.68(0.56)             & 14.3(0.8) & 52.3(1.7)              \\
                \bottomrule
            \end{tabular}
        \end{minipage}
    \end{subtable}

    \vspace{5mm}
    \begin{subtable}[t]{1.0\linewidth}
        \caption{Samson}
        \begin{minipage}[t]{0.26\linewidth}
            \vspace{0pt}
            \centering
            \begin{tabular}{ll}
                \toprule
                \multicolumn{2}{l}{\textbf{Baseline (LP method)}} \\
                \midrule
                MRSA & 3.34                                       \\
                Time & 31.2                                       \\
                \bottomrule
            \end{tabular}
        \end{minipage}
        \hfill
        \begin{minipage}[t]{0.71\linewidth}
            \vspace{0pt}
            \centering
            \begin{tabular}{
                    l
                    S[table-format=2.2(3)]
                    S[table-format=2.2(3)]
                    S[table-format=2.1(2)]
                    S[table-format=2.1(2)]
                }
                \toprule
                {$\lambda$}
                    & {MRSA ($\tau=1$)}
                    & {MRSA ($\tau=5$)}
                    & {Time ($\tau=1$)}
                    & {Time ($\tau=5$)}                                                               \\
                \midrule
                0   & 6.14              & \multicolumn{1}{c}{--} & 1.6       & \multicolumn{1}{c}{--} \\
                50  & 3.89(0.99)        & 3.72(0.48)             & 1.5(0.1)  & 1.7(0.0)               \\
                100 & 3.23(0.59)        & 3.05(0.28)             & 2.2(0.1)  & 5.9(0.5)               \\
                150 & 2.99(0.43)        & 2.90(0.25)             & 4.2(0.3)  & 15.7(0.8)              \\
                200 & 2.91(0.34)        & 2.77(0.22)             & 7.0(0.6)  & 29.5(1.1)              \\
                250 & 2.87(0.37)        & 2.68(0.22)             & 11.2(1.0) & 52.6(2.4)              \\
                \bottomrule
            \end{tabular}
        \end{minipage}
    \end{subtable}

    \vspace{5mm}
    \begin{subtable}[t]{1.0\linewidth}
        \caption{Urban}
        \begin{minipage}[t]{0.26\linewidth}
            \vspace{0pt}
            \centering
            \begin{tabular}{ll}
                \toprule
                \multicolumn{2}{l}{\textbf{Baseline (LP method)}} \\
                \midrule
                MRSA & 7.90                                       \\
                Time & 29395.5                                    \\
                \bottomrule
            \end{tabular}
        \end{minipage}
        \hfill
        \begin{minipage}[t]{0.71\linewidth}
            \vspace{0pt}
            \centering
            \begin{tabular}{
                    l
                    S[table-format=2.2(3)]
                    S[table-format=2.2(3)]
                    S[table-format=3.1(4)]
                    S[table-format=4.1(4)]
                }
                \toprule
                {$\lambda$}
                      & {MRSA ($\tau=1$)}
                      & {MRSA ($\tau=5$)}
                      & {Time ($\tau=1$)}
                      & {Time ($\tau=5$)}                                                                  \\
                \midrule
                0     & 15.86             & \multicolumn{1}{c}{--} & 126.5        & \multicolumn{1}{c}{--} \\
                2500  & 10.59(1.64)       & 9.29(0.97)             & 122.0(5.0)   & 209.1(10.5)            \\
                5000  & 9.74(1.66)        & 8.41(0.59)             & 162.9(26.9)  & 435.0(73.9)            \\
                7500  & 9.43(1.04)        & 8.16(0.48)             & 218.6(74.1)  & 765.7(244.4)           \\
                10000 & 9.24(0.91)        & 7.98(0.35)             & 275.8(99.8)  & 1009.0(270.7)          \\
                12500 & 9.14(0.73)        & 7.99(0.39)             & 361.6(175.5) & 1339.5(411.6)          \\
                \bottomrule
            \end{tabular}
        \end{minipage}
    \end{subtable}
\end{table}

We finally present the experimental results comparing REDIC with existing methods, namely,
SPA, ER, VCA, SNPA, MERIT, and MinVol.
The parameter settings for these methods, except for MinVol, were the same as those used in \cite{Miz25}
and are therefore omitted here.
However, we describe the setting of the parameter $\lambda$ in MERIT,
which appears in the optimization problem shown in~\eqref{eq: self-dictionary formulation for FGNSR and MERIT},
since the performance of MERIT is sensitive to the choice of $\lambda$.
We set $\lambda$ to 10 values: $10^{-6}$, $10^{-5}$, \ldots, $10^3$.
For each value of $\lambda$, we ran MERIT on each dataset and computed the corresponding MRSA score.
The best MRSA score among the 10 values of $\lambda$ was selected as the final MRSA score of MERIT on each dataset.
MinVol also involves a parameter $\lambda$ in the optimization problem shown in~\eqref{eq: min-vol NMF formulation}.
We used the same parameter selection procedure for $\lambda$ in MinVol as that used in MERIT,
via the parameter \texttt{lambda} in the MATLAB function \texttt{minvolNMF}.
The other parameters of \texttt{minvolNMF} were set as follows:
\texttt{maxiter} = 1000, \texttt{target} = 0.05, and \texttt{model} = 3.

Table~\ref{tab: experimental results by existing methods for the three datasets}
summarizes the MRSA scores of these methods for the three datasets.
From these results, we make the following observations.
For Jasper Ridge, MinVol achieves the lowest MRSA score among the existing methods,
while REDIC with $(\lambda, \tau) = (100, 1)$ achieves a lower MRSA score than MinVol.
For Samson, ER achieves the lowest MRSA score among the existing methods,
while REDIC with $(\lambda, \tau) = (250, 5)$ achieves a comparable MRSA score.
For Urban, MinVol achieves the lowest MRSA score among the existing methods,
while REDIC with $(\lambda, \tau) = (2500, 1)$ achieves a lower MRSA score than MinVol.

\begin{table}[h]
    \centering
    \footnotesize
    \caption{MRSA scores of the existing methods on the three datasets.}
    \label{tab: experimental results by existing methods for the three datasets}

    \begin{tabular}{l r r r r r r}
        \toprule
                     & SPA   & ER    & VCA   & SNPA  & MERIT & MinVol \\
        \midrule
        Jasper Ridge & 21.32 & 17.75 & 26.07 & 22.15 & 10.65 & 9.59   \\
        Samson       & 25.14 & 2.64  & 3.30  & 2.88  & 22.58 & 11.92  \\
        Urban        & 17.81 & 17.00 & 28.02 & 19.68 & 19.47 & 13.01  \\
        \bottomrule
    \end{tabular}

\end{table}

\begin{remark} \label{rem: performance of PSPA}
    In \cite{Miz25}, MATLAB code for PSPA was developed and its performance was evaluated using MATLAB R2021a.
    Using this code, we conducted experiments on the Urban dataset with MATLAB R2024b.
    The code yielded an MRSA score of 17.00 with MATLAB R2024b,
    whereas it yielded 12.99 with MATLAB R2021a.
    This discrepancy appears to be caused by differences in implementation behavior across MATLAB versions.
    Therefore, we did not include the PSPA results in this paper.
\end{remark}

\subsubsection{Ablation Studies} \label{subsubsec: ablation studies}
As shown in Section~\ref{subsubsec: comparison with existing methods},
REDIC achieves endmember extraction performance comparable to that of the LP method
when appropriate values of $\lambda$ and $\tau$ are used.
To examine whether this performance can be attributed to the output of DRS,
we conducted ablation studies in which Step~2 was modified to skip DRS, with $\KC$ set to $\emptyset$.
The parameter $r$ was set to the number of endmembers for each dataset.
The parameter $\tau$ was set to 1 and 5,
and $\lambda$ was varied from 50 to 250 in increments of 50 for Jasper Ridge and Samson,
and from 2500 to 12500 in increments of 2500 for Urban.
For each dataset, the modified REDIC was run 50 times,
and the mean and standard deviation of the MRSA scores and elapsed times were computed.

Table~\ref{tab: ablation results for the three datasets} summarizes the results.
From these results, we make the following observations:
\begin{itemize}
    \item For Samson, the modified REDIC yields lower MRSA scores than REDIC for all values of $\lambda$ and $\tau$.
          The gap between the two methods decreases as $\lambda$ increases, becoming less than 0.5 when $\lambda = 250$.
          We discuss the underlying cause below.
    \item For Jasper Ridge, the modified REDIC yields larger MRSA scores than REDIC when $\lambda \ge 100$, with consistently larger standard deviations.
    \item For Urban, the modified REDIC yields larger MRSA scores and larger standard deviations than REDIC for all $\lambda$ and $\tau$.
\end{itemize}

\begin{table}[h]
    \centering
    \footnotesize
    \caption{
        Ablation results on three datasets for different values of $\lambda$ and $\tau$.
        MRSA denotes the MRSA score. The mean and standard deviation over 50 runs are shown.
    }
    \label{tab: ablation results for the three datasets}

    \begin{subtable}[t]{0.48\linewidth}
        \centering
        \caption{Jasper Ridge}
        \begin{tabular}{
                l
                S[table-format=1.2(3)]
                S[table-format=1.2(3)]
            }
            \toprule
            {$\lambda$} & {MRSA ($\tau=1$)} & {MRSA ($\tau=5$)} \\
            \midrule
            50          & 8.48(4.25)        & 8.63(4.29)        \\
            100         & 8.69(4.13)        & 7.92(3.04)        \\
            150         & 8.11(3.17)        & 7.38(2.75)        \\
            200         & 7.96(3.55)        & 8.41(3.93)        \\
            250         & 8.18(3.08)        & 6.98(1.92)        \\
            \bottomrule
        \end{tabular}
    \end{subtable}
    \hfill
    \begin{subtable}[t]{0.48\linewidth}
        \centering
        \caption{Samson}
        \begin{tabular}{
                l
                S[table-format=1.2(3)]
                S[table-format=1.2(3)]
            }
            \toprule
            {$\lambda$} & {MRSA ($\tau=1$)} & {MRSA ($\tau=5$)} \\
            \midrule
            50          & 2.22(0.52)        & 1.82(0.20)        \\
            100         & 2.28(0.40)        & 2.01(0.22)        \\
            150         & 2.34(0.42)        & 2.06(0.23)        \\
            200         & 2.44(0.40)        & 2.15(0.17)        \\
            250         & 2.50(0.45)        & 2.21(0.18)        \\
            \bottomrule
        \end{tabular}
    \end{subtable}

    \vspace{5mm}
    \begin{subtable}[t]{0.48\linewidth}
        \centering
        \caption{Urban}
        \begin{tabular}{
                l
                S[table-format=2.2(3)]
                S[table-format=2.2(3)]
            }
            \toprule
            {$\lambda$} & {MRSA ($\tau=1$)} & {MRSA ($\tau=5$)} \\
            \midrule
            2500        & 16.71(3.63)       & 15.65(2.48)       \\
            5000        & 16.56(2.84)       & 15.50(1.81)       \\
            7500        & 16.04(3.03)       & 15.32(2.12)       \\
            10000       & 16.04(3.57)       & 14.52(2.16)       \\
            12500       & 16.60(3.51)       & 14.45(2.34)       \\
            \bottomrule
        \end{tabular}
    \end{subtable}

\end{table}

To explore the cause of the lower MRSA scores of the modified REDIC for Samson,
we examined the proportion of columns in an HSI matrix $A$
whose MRSA values with respect to each reference endmember signature $\w_j$ are below a threshold $\theta$.
Specifically, we constructed the column index set $\TC_j = \{ i \in [n] \mid \MRSA(\a_i, \w_j) \le \theta  \}$ for each $j \in [r]$,
and evaluated the proportion of these columns as $|\TC_j| / n$.
The values of $\theta$ were chosen based on the MRSA scores of the LP method.

Table~\ref{tab: proportion of columns within an MRSA value of theta in three datasets}
summarizes the results for the three datasets.
These results show that, for Samson, if we randomly select a few hundred columns from the HSI matrix,
a few dozen of them have MRSA values below 4 with respect to each reference endmember signature.
In contrast, for Jasper Ridge,
although some similar columns may exist for the endmember corresponding to Road,
their number is only a few.
A similar observation holds for Urban.
Among thousands of randomly selected columns,
very few are similar to the endmembers corresponding to Roof~1 and Roof~2.

REDIC solves model $\HT$ with RCE
and then performs a postprocessing step to obtain the estimates of the endmember signatures.
In this postprocessing step, clusters of columns are constructed,
and one column is selected from each cluster as an estimated endmember signature.
Accordingly, if the input HSI matrix contains a large number of columns similar to each endmember signature,
this step is expected to yield more accurate estimates.

From these observations,
the lower MRSA scores of the modified REDIC for Samson can be attributed to
the presence of relatively many columns similar to each reference endmember signature in the HSI matrix.
Therefore, even when DRS is replaced with random column selection,
the modified REDIC can still extract endmembers with high accuracy.
In contrast,
incorporating DRS into REDIC is necessary to achieve high endmember extraction performance
for HSI matrices in which only a small number of similar columns are present for some endmembers.
\begin{table}[h]
    \centering
    \footnotesize
    \caption{
        Proportion of columns whose MRSA values are at most $\theta$
        for each reference endmember signature in three datasets.
    }
    \label{tab: proportion of columns within an MRSA value of theta in three datasets}
    \begin{subtable}[t]{0.48\linewidth}
        \centering
        \caption{Jasper Ridge}
        \begin{tabular}{lrrrr}
            \toprule
            $\theta$ & Tree    & Water   & Soil   & Road   \\
            \midrule
            7        & 23.7 \% & 27.2 \% & 8.7 \% & 2.3 \% \\
            6        & 20.9 \% & 23.4 \% & 6.7 \% & 1.5 \% \\
            5        & 17.7 \% & 14.3 \% & 4.9 \% & 0.9 \% \\
            \bottomrule
        \end{tabular}
    \end{subtable}
    \hfill
    \begin{subtable}[t]{0.48\linewidth}
        \centering
        \caption{Samson}
        \begin{tabular}{lrrr}
            \toprule
            $\theta$ & Soil    & Tree    & Water   \\
            \midrule
            4        & 23.5 \% & 27.0 \% & 10.0 \% \\
            3        & 22.2 \% & 19.8 \% & 4.7 \%  \\
            2        & 19.3 \% & 12.3 \% & 0.1 \%  \\
            \bottomrule
        \end{tabular}
    \end{subtable}

    \vspace{5mm}
    \begin{subtable}[t]{0.7\linewidth}
        \centering
        \caption{Urban}
        \begin{tabular}{lrrrrrr}
            \toprule
            $\theta$ & Asphalt & Grass   & Tree    & Roof~1 & Roof~2 & Soil   \\
            \midrule
            8        & 6.4 \%  & 31.3 \% & 27.0 \% & 0.3 \% & 0.1 \% & 6.6 \% \\
            7        & 4.5 \%  & 27.1 \% & 23.5 \% & 0.3 \% & 0.1 \% & 4.4 \% \\
            6        & 2.6 \%  & 22.8 \% & 19.8 \% & 0.2 \% & 0.1 \% & 2.7 \% \\
            \bottomrule
        \end{tabular}
    \end{subtable}
\end{table}

\section{Concluding Remarks} \label{sec: concluding remarks}
We studied a data reduction technique that removes redundant pixels for endmember extraction from HSIs.
This work was motivated by the high computational cost of self-dictionary methods when applied to large-scale HSIs.
We conducted a theoretical analysis of this reduction step
and developed a data reduction algorithm based on a splitting technique, called DRS.
By combining DRS with a self-dictionary method that uses an LP formulation,
we proposed a data-reduced self-dictionary method, REDIC, for endmember extraction.
Numerical experiments demonstrated the effectiveness of REDIC in terms of both computational time and extraction accuracy.

We conclude this paper by suggesting directions for future research.
As observed in Section~\ref{subsec: performance evaluation of REDIC},
there are noticeable gaps between the MRSA distances of the outputs produced by DRS and the MRSA scores obtained by REDIC with $\lambda = 0$,
indicating that the extraction accuracy of REDIC still leaves room for improvement.
One possible approach to addressing this limitation is to use a combinatorial method in place of the LP method within REDIC.
Gillis \cite{Gil19} studied a combinatorial approach for a slightly generalized separable NMF problem and
showed theoretically that it exhibits strong robustness to noise.
This suggests that combinatorial methods may achieve superior endmember extraction accuracy.
Due to their high computational cost,
applying such methods directly to endmember extraction is impractical.
However, by integrating them with DRS,
it may be possible to achieve higher accuracy while maintaining reasonable computational time,
potentially outperforming the current version of REDIC.

As mentioned in Section~\ref{subsec: problem formulation},
the endmember extraction problem is closely related to problems such as topic modeling and community detection.
It would be interesting to investigate
whether the proposed data reduction technique can also be applied to these problems
and thereby accelerate progress in these research areas.

\appendix
\section{Analysis of Reduced Hyperspectral Image Data} \label{sec: analysis of reduced hyperspectral image data}
Throughout this section, we assume that $A \in \Real^{d \times n}$ is a nearly $r$-separable matrix of
the form $A = M + V$, where $M \in \Real^{d \times n}_+$ is $r$-separable and can be written as $M = WH$,
as given in \eqref{eq: separable matrix}, and $V \in \Real^{d \times n}$ is noise.
Before presenting the proof of Theorem~\ref{thm: main result},
we establish some properties of the parameter $\rho$.

\begin{lemma} \label{lem: linear independence and tau}
    The following statements hold for $Z \in \Real^{m \times n}$:
    \begin{enumerate}[label={\normalfont(\alph*)}]
        \item $\rho(Z) = 0$ if and only if the $n$ columns of $Z$ are linearly dependent.
        \item $\rho(Z) > 0$ if and only if the $n$ columns of $Z$ are linearly independent.
    \end{enumerate}
\end{lemma}
\begin{proof}
    Since part (b) is immediate from part (a), we prove only part (a).
    We first show the ``only if'' direction.
    There exists $\x \in \Real^n$ such that $Z\x = \zero$ and $\| \x \|_1 = 1$,
    which implies that the $n$ columns of $Z$ are linearly dependent.
    We next prove the ``if'' direction. There exists $\x \neq \zero$ such that $Z\x = \zero$.
    Let $\bar{\x} = \x / \| \x \|_1$,
    Then $Z \bar{\x} = \zero$ and $\| \bar{\x} \|_1 = 1$. Accordingly, $\rho(Z) = 0$.
\end{proof}

\begin{lemma} \label{lem: upper bound on tau}
    Let $Z \in \Real^{m \times n}$ satisfy $\| Z \|_1 = 1$. Then $\rho(Z) \le 1$ holds.
\end{lemma}
\begin{proof}
    For any $\x \in \Real^n$ such that $\| \x \|_1 = 1$, we have $\| Z\x \|_1 \le \| Z \|_1 \| \x \|_1 = 1$.
    Accordingly, $\rho(Z) = \min_{\|\x\|_1 = 1} \| Z\x \|_1 \le 1$.
\end{proof}
As an immediate consequence of Lemma \ref{lem: linear independence and tau}, we obtain the following result:
\begin{corollary} \label{cor: m ge n}
    If $\rho(Z) > 0$ for $Z \in \Real^{m \times n}$, then $m \ge n$.
\end{corollary}
\begin{proof}
    If $\rho(Z) > 0$, then
    by part (b)~of Lemma~\ref{lem: linear independence and tau}, the $n$ columns of $Z$ are linearly independent,
    which implies that $\rank(Z) = n$.
    Since $\rank(Z) = \min\{m,n\} \le m$, it follows that $m \ge n$.
\end{proof}

\subsection{Proof of Theorem~\ref{thm: main result}}
To prove Theorem~\ref{thm: main result},
we establish two propositions: Propositions~\ref{prop: framwork for proving the main theorem} and~\ref{prop: evaluation of mu(j)}.
The theorem then follows immediately from these propositions.
In the first proposition, we focus on the factor $H$ of the $r$-separable matrix $M = WH$ contained in $A$.
We choose $\KC$ from $\Gamma(A)$ and consider, for each $j$, the maximum value in the $j$th row of $H(\KC)$
and the index at which this maximum is attained.
We denote these respectively by $\mu(j)$ and $k(j)$. More precisely, given $\KC \in \Gamma(A)$, define
\begin{align}
    \mu(j) & = \max_{k \in \KC} H(j,k),     \label{eq: mu(j)} \\
    k(j)   & = \arg \max_{k \in \KC} H(j,k) \label{eq: k(j)}
\end{align}
for each $j \in [r]$.
The following proposition provides a framework for proving Theorem~\ref{thm: main result}.
\begin{proposition} \label{prop: framwork for proving the main theorem}
    Let the nearly $r$-separable matrix $A = M + V$ with $M=WH$ satisfy Assumption \ref{asm: nearly r-separable matrix}.
    Let $\KC \in \Gamma(A)$.
    If $\epsilon < \rho(W)$ and $\mu(j) > 1/2$ for each $j \in [r]$,
    then
    \begin{itemize}
        \item the elements $k(1), \ldots, k(r)$ of $\KC$ are all distinct, and
        \item $\| \w_j - \a_{k(j)} \|_1 \le 2(1 - \mu(j)) + \epsilon$ for each $j \in [r]$.
    \end{itemize}
\end{proposition}
The proof is given in Section \ref{subsec: proof of the first proposition}.
It is established using Lemmas~\ref{lem: bound on the distance between w_j and a_i}--\ref{lem: k is injective}.
Lemma~\ref{lem: bound on the distance between w_j and a_i} provides an upper bound on the $L_1$ distance
between $\w_j$  and $\a_{k(j)}$ for each $j \in [r]$,
while Lemmas~\ref{lem: size of K} and~\ref{lem: k is injective} establish that the elements $k(1), \ldots, k(r)$ of $\KC$ are all distinct.
Next, we show in the second proposition that, if $\epsilon$ is below a certain threshold, then $\mu(j)$ exceeds $1/2$.
\begin{proposition} \label{prop: evaluation of mu(j)}
    Let the nearly $r$-separable matrix $A = M + V$ with $M = WH$ satisfy Assumption \ref{asm: nearly r-separable matrix}.
    Let $\KC \in \Gamma(A)$.
    If $\epsilon < \rho(W)/9$, then
    \begin{align*}
        1 - \mu(j) \le \frac{4 \epsilon}{\rho(W)(1 - \epsilon)} < \frac{1}{2}
    \end{align*}
    for each $j \in [r]$.
\end{proposition}
The proof is given in Section \ref{subsec: proof of the first proposition}.
It is established using Lemmas~\ref{lem: upper bound on mu(j)} and \ref{lem: size of the norm of c}.
We now prove Theorem \ref{thm: main result} by combining Propositions \ref{prop: framwork for proving the main theorem}
and \ref{prop: evaluation of mu(j)}.
\begin{proof}[(Proof of Theorem \ref{thm: main result})]
    It follows from Propositions \ref{prop: framwork for proving the main theorem} and \ref{prop: evaluation of mu(j)}
    that $\KC \in \Gamma(A)$ contains $r$ distinct indices $k_1, \ldots, k_r$, which satisfy
    \begin{align*}
        \| \w_j - \a_{k_j} \|_1 \le 2(1 - \mu(j)) + \epsilon \le  \frac{8}{\rho(W)} \frac{\epsilon}{1 - \epsilon} + \epsilon
    \end{align*}
    for each $j \in [r]$.
    Lemma \ref{lem: upper bound on tau}, together with Assumption \ref{asm: nearly r-separable matrix}(a), implies that $\rho(W) \le 1$.
    Therefore, $0 \le \epsilon < \rho(W)/9 \le 1/9$.
    It then follows that $\epsilon / ( 1 - \epsilon) < 9\epsilon / 8$.
    Consequently, we obtain
    $\| \w_j - \a_{k_j} \|_1 < (9 / \rho(W) + 1) \epsilon$ for each $j \in [r]$.
\end{proof}

\subsection{Proof of Proposition \ref{prop: framwork for proving the main theorem}}
\label{subsec: proof of the first proposition}
We begin by evaluating the $L_1$ distance between $\w_j$ and $\a_i$.
We first show that this distance can be upper-bounded using the $(j,i)$th element of $H$
and the noise level $\epsilon$ under Assumption~\ref{asm: nearly r-separable matrix}.
\begin{lemma} \label{lem: bound on the distance between w_j and a_i}
    Let the nearly $r$-separable matrix $A = M + V$, with $M = WH$, satisfy Assumption \ref{asm: nearly r-separable matrix}.
    Then $\| \w_j - \a_i \|_1 \le 2 (1 - H(j,i)) + \epsilon$ holds for each $i \in [n]$ and each $j \in [r]$.
\end{lemma}
\begin{proof}
    The $i$th column $\a_i$ of $A$ can be expressed as $\a_i = \sum_{u=1}^r H(u,i) \w_u + \v_i$.
    In addition, by $H \ge O$ and Assumption \ref{asm: nearly r-separable matrix}(a),
    we have
    \begin{align*}
        \sum_{u=1}^{r} H(u,i) = 1 \equivSym 1 - H(j,i) = \sum_{u \neq j} H(u,i) \ge 0.
    \end{align*}
    Accordingly,  we can bound $\| \w_j - \a_i \|_1$ from above as follows:
    \begin{align*}
        \| \w_j - \a_i \|_1
         & = \| \w_j - \sum_{u=1}^r H(u,i) \w_u - \v_i \|_1                                   \\
         & = \| (1 - H(j,i)) \w_j - \sum_{u \neq j} H(u,i) \w_u - \v_i \|_1                   \\
         & \le (1 - H(j,i)) \| \w_j \|_1 + \sum_{u \neq j} H(u,i) \| \w_u \|_1 + \| \v_i \|_1 \\
         & \le 1 - H(j,i) + \sum_{u \neq j} H(u,i) + \epsilon                                 \\
         & = 2(1 - H(j,i)) + \epsilon
    \end{align*}
    where the first inequality follows from $H \ge O$ and $ 1 - H(j,i) \ge 0$ as shown above, and
    the second inequality follows from Assumption \ref{asm: nearly r-separable matrix}.
\end{proof}
Let $\KC \in \Gamma(A)$. This lemma implies that, for each $k \in \KC$,
the $L_1$ distance between $\w_j$ and  $\a_k$ satisfies $\| \w_j - \a_k \|_1 \le 2 (1 - H(j,k)) + \epsilon$.
In particular, this upper bound is minimized when $k = k(j)$, where $k(j)$ is defined as in \eqref{eq: k(j)}.
Hence, the column $\a_{k(j)}$ is located near $\w_j$, and the $L_1$ distance between them is bounded as
\begin{align*}
    \| \w_j - \a_{k(j)} \|_1 \le 2 (1 - \mu(j)) + \epsilon
\end{align*}
for each $j \in [r]$.

We next show that if $\mu(j) > 1/2$ for all $j \in [r]$, then the indices $k(1), \ldots, k(r)$ are all distinct.
In what follows, we regard $k(j)$ as a map from $[r]$ to $\KC$.
We show in Lemma~\ref{lem: size of K} that any $\KC \in \Gamma(A)$ contains at least $r$ elements.
Hence, it is possible for an injective map from $[r]$ to $\KC$ to exist.
We then show in Lemma~\ref{lem: k is injective} that if $\mu(j) > 1/2$ for all $j \in [r]$, then the map $k : [r] \rightarrow \KC$ is in fact injective.

\begin{lemma} \label{lem: size of K}
    Let the nearly $r$-separable matrix $A = M + V$ with $M=WH$ satisfy Assumption~\ref{asm: nearly r-separable matrix}.
    If $\epsilon < \rho(W)$, then any $\KC \in \Gamma(A)$ satisfies $|\KC| \ge r$.
\end{lemma}
\begin{proof}
    Since $\Gamma(A) \neq \emptyset$, there exists $\KC_* \in \Gamma(A)$
    such that $| \KC_* | \le | \KC |$ for any $\KC \in \Gamma(A)$.
    Hence, it suffices to show that $| \KC_* | \ge r$.
    We prove this inequality by contradiction.
    Suppose, contrary to the claim, that $| \KC_* | < r$.
    There are the columns $\a_{i_1}, \ldots, \a_{i_r}$ of $A$
    such that $\a_{i_j} = \w_j + \v_{i_j}$ for $j \in [r]$.
    By letting $\IC = \{i_1, \ldots, i_r\}$, we express these columns in a matrix form as $A(\IC) = W + V(\IC)$.
    For any $\x \in \Real^{r}$ with $\| \x \|_1 = 1$,
    \begin{align*}
        \| A(\IC) \x \|_1 = \| W\x + V(\IC) \x \|_1 \ge \| W\x \|_1 - \| V(\IC) \x \|_1 \ge \rho(W) - \epsilon.
    \end{align*}
    The last inequality follows from the definition of $\rho(W)$ and Assumption~\ref{asm: nearly r-separable matrix}(b).
    We thus get
    \begin{align*}
        \rho(A(\IC)) = \min_{\| \x \|_1 = 1} \| A(\IC) \x \|_1 \ge \rho(W) - \epsilon > 0.
    \end{align*}
    Taking into account Lemma~\ref{lem: linear independence and tau}(b), we obtain $\rank(A(\IC)) = r$.
    Meanwhile, since $\a_1, \ldots, \a_n \in \cone(A) = \cone A((\KC_*))$,
    we can express $A(\IC)$ as $A(\IC) = A(\KC_*) X$ using $X \in \Real^{|\KC_*| \times r}_+$, which implies that
    \begin{align*}
        \rank(A(\IC)) = \rank(A(\KC_*)X) \le \rank(A(\KC_*)) \le \min\{d,|\KC_*|\} = |\KC_*|.
    \end{align*}
    The last equality follows from the assumption $|\KC_*| < r$ and
    from $r \le d$, which holds by Corollary~\ref{cor: m ge n} together with $0 \le \epsilon < \rho(W)$.
    Thus we obtain the contradiction
    \begin{align*}
        \rank(A(\IC)) = r > |\KC_*| \ge \rank(A(\IC)).
    \end{align*}
    Therefore, $|\KC_*| \ge r$ must hold.
\end{proof}

\begin{lemma} \label{lem: k is injective}
    Let the nearly $r$-separable matrix $A = M + V$ with $M = WH$ satisfy Assumption~\ref{asm: nearly r-separable matrix}.
    Let $\KC \in \Gamma(A)$.
    If $\epsilon < \rho(W)$ and $\mu(j) > 1/2$ for each $j \in [r]$, then the map $k : [r] \rightarrow \KC$ is injective.
\end{lemma}
\begin{proof}
    Since $\epsilon < \rho(W)$, Lemma~\ref{lem: size of K} implies that an injective map from $[r]$ to $\KC$ can exist.
    We prove by contradiction that if $j \neq j'$, then $k(j) \neq k(j')$.
    Suppose, contrary to the claim, that $k = k(j) = k(j')$ for some $j \neq j'$.
    Since $\mu(j) > 1/2$ for each $j \in [r]$,
    this assumption implies that
    \begin{align*}
        H(j, k) = H(j, k(j)) = \mu(j) > 1/2,
        \qquad
        H(j', k) = H(j', k(j')) = \mu(j') > 1/2.
    \end{align*}
    Accordingly, since $H \ge O$, we obtain $\sum_{u = 1}^{r} H(u, k) > 1$.
    On the other hand, by Assumption \ref{asm: nearly r-separable matrix}(a) and $H \ge O$,
    we have  $\sum_{u = 1}^{r} H(u, k) = 1$.
    This yields the contradiction
    \begin{align*}
        1 =  \sum_{u = 1}^{r} H(u, k)  > 1.
    \end{align*}
    Therefore, the assumption is false, and hence $k(j) \neq k(j')$ whenever $j \neq j'$.
\end{proof}

We now prove Proposition~\ref{prop: framwork for proving the main theorem}.
\begin{proof}[(Proof of Proposition~\ref{prop: framwork for proving the main theorem})]
    The result follows from Lemmas~\ref{lem: bound on the distance between w_j and a_i} and~\ref{lem: k is injective}.
\end{proof}

\subsection{Proof of Proposition \ref{prop: evaluation of mu(j)}}
\label{subsec: proof of the second proposition}
We estimate an upper bound on $1- \mu(j)$.
The nearly $r$-separable matrix $A$ contains columns $\a_{i_1}, \ldots, \a_{i_r}$
such that $\a_{i_j} = \w_j + \v_{i_j}$ for $j \in [r]$.
Let $\KC  \in \Gamma(A)$.
Since $\a_{i_j} = \w_j + \v_{i_j} \in \cone(A) = \cone(A(\KC))$, there exists some $\c_j \in \Real^{|\KC|}$ such that
\begin{align} \label{eq: conical hull representation of perturbed generators}
    \w_j + \v_{i_j} = A(\KC) \c_j \quad \text{and} \quad \c_j \ge \zero
\end{align}
for each $j \in [r]$.
Using this representation, we derive an upper bound on $1 - \mu(j)$
in Lemmas~\ref{lem: upper bound on mu(j)} and~\ref{lem: size of the norm of c}.

\begin{lemma} \label{lem: upper bound on mu(j)}
    Let the nearly $r$-separable matrix $A = M + V$ with $M=WH$ satisfy Assumption~\ref{asm: nearly r-separable matrix}.
    Let $\KC \in \Gamma(A)$.
    If $\rho(W) > 0$, then for each $j \in [r]$, the following inequality holds:
    \begin{align*}
        1 - \mu(j) \le \frac{2 \epsilon + (1 + \epsilon) \left| \| \c_j \|_1 -1 \right|}{\rho(W)},
    \end{align*}
    where $\c_j$ is defined as in~\eqref{eq: conical hull representation of perturbed generators}.
\end{lemma}
\begin{proof}
    Following the notation stated in Section~\ref{sec: introduction},
    we write $\a(i)$ for the $i$th entry of a vector $\a$ in this proof.
    Let $\KC$ contain $s$ elements $k_1, \ldots, k_s$.
    For $\c_j$, let $\bar{\c}_j = \c_j / \| \c_j \|_1$.
    We rewrite
    \begin{align*}
        \w_j + \v_{i_j} = A(\KC)\bar{\c}_j + A(\KC) (\c_j - \bar{\c}_j).
    \end{align*}
    Each column $\a_{k_u}$ of $A(\KC)$ can be expressed as
    $\a_{k_u} = W\h_{k_u} + v_{k_u}$ for $u \in [s]$.
    Hence,
    \begin{align*}
        A(\KC) \bar{\c}_j = \sum_{u=1}^{s} \bar{\c}_j(u) \a_{k_u}
        = \sum_{u=1}^{s} \bar{\c}_j(u) (W \h_{k_u} + \v_{k_u}) =  W \p + \sum_{u=1}^{s} \bar{\c}_j(u) \v_{k_u}.
    \end{align*}
    where $\p = \sum_{u=1}^{s} \bar{\c}_j(u) \h_{k_u} \in \Real^r$.
    Thus,
    \begin{align*}
        \w_j + \v_{i_j} = W\p + \sum_{u=1}^{s} \bar{\c}_j(u) \v_{k_u} +  A(\KC) (\c_j - \bar{\c}_j)
        \equivSym \y = W\q
    \end{align*}
    where $\y = \v_{i_j} - \sum_{u=1}^{s} \bar{\c}_j(u) \v_{k_u} - A(\KC) (\c_j - \bar{\c}_j) \in \Real^d$
    and $\q = \p - \e_j \in \Real^r$ where $\e_j$ denotes the $j$th unit vector in $\Real^r$.

    \medskip \noindent
    \textbf{Case of $\q \neq \zero$.} \
    Let $\bar{\q} = \q / \| \q \|_1$ for $\q$.
    From the equality $\y = W \q$, we obtain
    \begin{align*}
        \| \y \|_1 = \| W \q \|_1 = \| \q \|_1  \| W \bar{\q} \|_1 \ge \rho(W) \| \q \|_1.
    \end{align*}
    Since $\rho(W) > 0$, it follows that
    \begin{align*}
        \| \y \|_1 / \rho(W) \ge \| \q \|_1 = \sum_{u=1}^{r} |\q(u)| \ge | \p(j) - 1 |.
    \end{align*}
    Therefore,
    \begin{align*}
        \p(j) \ge 1 - \| \y \|_1 / \rho(W)
    \end{align*}
    for $j \in [r]$.
    Recall that $\mu(j)$ is defined as $\mu(j) = \max_{k \in \KC} H(j,k)$ in~\eqref{eq: mu(j)}.
    Since the columns of $H(\KC)$ are $\h_{k_1}, \ldots, \h_{k_s}$,
    we can equivalently write
    \begin{align*}
        \mu(j) = \max_{u = 1,\ldots,s} \h_{k_u}(j).
    \end{align*}
    Using this representation, we observe that $\p(j)$ is upper-bounded by $\mu(j)$:
    \begin{align} \label{ineq: p(j) is upper-bounded by mu(j)}
        \p(j) = \sum_{u=1}^{s} \bar{\c}_j(u) \h_{k_u}(j) \le \sum_{u=1}^{s} \bar{\c}_j(u) \mu(j) = \mu(j)
    \end{align}
    where the last equality uses $\| \bar{\c} \|_1 = 1$ and $\bar{\c} \ge \zero$.
    Hence,
    \begin{align} \label{ineq: bound on 1 - mu(j) using y}
        1 - \mu(j) \le \| \y \|_1 / \rho(W)
    \end{align}
    for $j \in [r]$.
    We now bound the $L_1$ norm of $\y$:
    \begin{align} \label{ineq: bound on the norm of y}
        \| \y \|_1
         & = \| \v_{i_j} - \sum_{u=1}^{s} \bar{\c}_j(u) \v_{k_u} - A(\KC) (\c_j - \bar{\c}_j) \|_1 \notag                           \\
         & \le \| \v_{i_j} \|_1 + \sum_{u=1}^{s} |\bar{\c}_j(u)| \| \v_{k_u} \|_1 + \| A(\KC) \|_1 \| \c_j - \bar{\c}_j \|_1 \notag \\
         & \le 2 \epsilon + \| A(\KC) \|_1 \| \c_j - \bar{\c}_j \|_1 \notag                                                         \\
         & \le 2 \epsilon + ( 1 + \epsilon ) \left| \| \c_j \|_1 - 1 \right|.
    \end{align}
    The last inequality uses
    \begin{align} \label{ineq: bound on the norm of A(KC)}
        \| A(\KC) \|_1 = \| WH(\KC) + V(\KC) \|_1 \le \| W \|_1 \| H(\KC) \|_1 + \| V(\KC) \|_1 \le 1 + \epsilon,
    \end{align}
    from Assumption~\ref{asm: nearly r-separable matrix}, and
    \begin{align*}
        \| \c_j - \bar{\c}_j \|_1 = \| \| \c_j \|_1 \bar{\c}_j - \bar{\c}_j \|_1
        = \| \left( \| \c_j \|_1 - 1 \right) \bar{\c}_j \|_1
        = \left| \| \c_j \|_1 - 1 \right| \| \bar{\c}_j \|_1 = \left| \| \c_j \|_1 - 1 \right|.
    \end{align*}
    Combining inequalities~\eqref{ineq: bound on 1 - mu(j) using y} and \eqref{ineq: bound on the norm of y} gives the desired result.

    \medskip \noindent
    \textbf{Case of $\q = \zero$.} \
    In this case, we have $\p = \e_j$.
    Moreover, by~\eqref{ineq: p(j) is upper-bounded by mu(j)},
    we have $\mu(j) \ge \p(j)$.
    Thus,
    \begin{align*}
        1 = \p(j) \le \mu(j) \equivSym 1 - \mu(j) \le 0.
    \end{align*}
    Since
    \begin{align*}
        \frac{2 \epsilon + (1 + \epsilon) \left| \| \c_j \|_1 -1 \right|}{\rho(W)} \ge 0,
    \end{align*}
    we obtain the desired result.
\end{proof}

\begin{lemma} \label{lem: size of the norm of c}
    Let the nearly $r$-separable matrix $A = M + V$ with $M=WH$ satisfy Assumption~\ref{asm: nearly r-separable matrix}.
    Let $\KC \in \Gamma(A)$.
    If $\epsilon < 1$, then for each $j \in [r]$,
    \begin{align*}
        \left| \| \c_j \|_1 - 1 \right| \le \frac{2 \epsilon}{ 1 - \epsilon},
    \end{align*}
    where $\c_j$ is defined in~\eqref{eq: conical hull representation of perturbed generators}.
\end{lemma}
\begin{proof}
    From the representation $\w_j + \v_{i_j} = A(\KC) \c_j$ with $\c_j \ge \zero$ in~\eqref{eq: conical hull representation of perturbed generators},
    we obtain the following bounds:
    \begin{align}
         & 1 - \epsilon \le \| \w_j + \v_{i_j} \|_1 = \| A(\KC) \c_j \|_1 \le (1 + \epsilon) \| \c_j \|_1,
        \label{ineq: (1-epsilon <= (1+epsilon)||c_j|| )}                                                   \\
         & 1 + \epsilon \ge \| \w_j + \v_{i_j} \|_1 = \| A(\KC) \c_j \|_1 \ge (1 - \epsilon) \| \c_j \|_1.
        \label{ineq: (1-epsilon >= (1+epsilon)||c_j|| )}
    \end{align}
    We justify the bounds in~\eqref{ineq: (1-epsilon <= (1+epsilon)||c_j|| )}.
    The left inequality follows from
    \begin{align*}
        \| \w_j + \v_{i_j} \|_1 \ge \| \w_j \|_1 - \| \v_{i_j} \|_1 \ge 1 - \epsilon,
    \end{align*}
    and the right inequality follows from
    \begin{align*}
        \| A(\KC) \c_j \|_1 \le \| A(\KC) \|_1 \| \c_j \|_1 \le (1 + \epsilon) \| \c_j \|_1
    \end{align*}
    where the final inequality uses \eqref{ineq: bound on the norm of A(KC)} from the proof of Lemma~\ref{lem: upper bound on mu(j)}.
    We justify the bounds in~\eqref{ineq: (1-epsilon >= (1+epsilon)||c_j|| )}.
    The left inequality follows from
    \begin{align*}
        \| \w_j + \v_{i_j} \|_1 \le \| \w_j \|_1 + \| \v_{i_j} \|_1 \le 1 + \epsilon,
    \end{align*}
    and the right inequality follows from
    \begin{align*}
        \| A(\KC) \c_j \|_1
         & = \| (WH(\KC) + V(\KC)) \c_j \|_1           \\
         & \ge \| WH(\KC) \c_j \|_1 - \|V(\KC)\c_j\|_1 \\
         & = \|\c_j \|_1 - \|V(\KC)\c_j\|_1            \\
         & \ge \|\c_j \|_1 - \|V(\KC)\|_1 \| \c_j\|_1  \\
         & \ge (1-\epsilon) \| \c_j \|_1
    \end{align*}
    where the second equality follows from the nonnegativity of $W$, $H$, and $\c_j$,
    together with Assumption~\ref{asm: nearly r-separable matrix}(a).
    Since $\epsilon < 1$,
    inequalities~\eqref{ineq: (1-epsilon <= (1+epsilon)||c_j|| )} and~\eqref{ineq: (1-epsilon >= (1+epsilon)||c_j|| )} imply
    that
    \begin{align*}
        \frac{1 - \epsilon}{1 + \epsilon} \le \| \c_j \|_1 \le \frac{1 + \epsilon}{1 - \epsilon}
        \equivSym
        \frac{-2\epsilon}{1 + \epsilon} \le \| \c_j \|_1 - 1 \le \frac{2\epsilon}{1 - \epsilon}.
    \end{align*}
    Since
    \begin{align*}
        0 \le \frac{2\epsilon}{1 + \epsilon} \le \frac{2\epsilon}{1 - \epsilon},
    \end{align*}
    we conclude that
    \begin{align*}
        \frac{-2\epsilon}{1 - \epsilon} \le \| \c_j \|_1 - 1 \le \frac{2\epsilon}{1 - \epsilon}
        \equivSym
        \left| \| \c_j \|_1 - 1 \right| \le \frac{2\epsilon}{1 - \epsilon},
    \end{align*}
    completing the proof.
\end{proof}

We now prove Proposition~\ref{prop: evaluation of mu(j)}.
\begin{proof}[(Proof of Proposition~\ref{prop: evaluation of mu(j)})]
    By definition, $\epsilon \ge 0$.
    Lemma~\ref{lem: upper bound on tau}, together with Assumption~\ref{asm: nearly r-separable matrix}(a), implies that $\rho(W) \le 1$.
    Since the proposition assumes that $\epsilon < \rho(W) / 9$,
    we have
    \begin{align*}
        0 \le \epsilon < \rho(W) / 9 \le 1/9 < 1.
    \end{align*}
    Applying Lemmas~\ref{lem: upper bound on mu(j)} and~\ref{lem: size of the norm of c} yields
    \begin{align*}
        1 - \mu(j) \le \frac{4 \epsilon}{\rho(W) (1 - \epsilon)}
    \end{align*}
    for each $j \in [r]$. Moreover, since $\epsilon < \rho(W) / 9 \le 1/9$,
    we obtain
    \begin{align*}
        1 - \mu(j) \le \frac{4 \epsilon}{\rho(W) (1 - \epsilon)} < \frac{4 \rho(W) / 9 }{8 \rho(W) / 9} = \frac{1}{2}.
    \end{align*}
    Hence, $1 - \mu(j) < 1/2$ for each $j \in [r]$.
\end{proof}

\section{Analysis of the Outputs Produced by DR and DRS} \label{sec: analysis of DR and DRS}
We summarize the result concerning the outputs produced by DR and DRS in the following theorem.
\begin{theorem} \label{thm: analysis of DR and DRS outputs}
    Let $\KC$ be the output produced by either DR or DRS for the input matrix $A$. Then the following statements hold:
    \begin{enumerate}[label={\normalfont(\alph*)}]
        \item $\KC \in \Gamma(A)$.
        \item For any $k \in \KC$, the set $\KC - k$ does not belong to $\Gamma(A)$.
    \end{enumerate}
\end{theorem}
To prove this theorem for the case of DRS, we establish the following lemma.
\begin{lemma} \label{lem: cone equality over partition}
    Let $\{\IC_1, \ldots, \IC_p\}$ be a $p$-way partition of $[n]$.
    If $\KC_j \subset \IC_j$ satisfies $\cone(A(\IC_j)) = \cone(A(\KC_j))$ for each $j \in [p]$,
    then $\cone(A) = \cone(A(\KC_1 \cup \cdots \cup \KC_p))$.
\end{lemma}
\begin{proof}
    It suffices to show the inclusion $\cone(A) \subset \cone(A(\KC_1 \cup \cdots \cup \KC_p))$,
    since the reverse inclusion follows immediately from  $[n] \supset \KC_1 \cup \cdots \cup \KC_p$.
    Let $\a \in \cone(A)$. Then
    \begin{align*}
        \a = \sum_{i \in [n]} \alpha_i \a_i
    \end{align*}
    for some coefficients $\alpha_i \ge 0$.
    Since $\{\IC_1, \ldots, \IC_p\}$ is a $p$-way partition of $[n]$, this can be rewritten as
    \begin{align*}
        \a = \sum_{j \in [p]} \sum_{i \in \IC_j} \alpha_i \a_i.
    \end{align*}
    For each $j \in [p]$, define $\b_j = \sum_{i \in \IC_j} \alpha_i \a_i$.
    Then $\b_j \in \cone(A(\IC_j)) = \cone(A(\KC_j))$ by assumption.
    Thus, $\b_j$ can be written as
    \begin{align*}
        \b_j = \sum_{k \in \KC_j} \beta_k \a_k
    \end{align*}
    for some coefficients $\beta_k \ge 0$.
    Combining these, we obtain
    \begin{align*}
        \a = \sum_{j \in [p]}  \sum_{k \in \KC_j} \beta_k \a_k.
    \end{align*}
    Because $\IC_{j_1} \cap \IC_{j_2} = \emptyset$ for $j_1 \neq j_2$
    and $\KC_j \subset \IC_j$, it follows that
    $\KC_{j_1} \cap \KC_{j_2} = \emptyset$ as well.
    Hence we can rewrite the above expression simply as
    \begin{align*}
        \a = \sum_{k \in \KC_1 \cup \cdots \cup \KC_p} \beta_k \a_k.
    \end{align*}
    This proves the desired inclusion.
\end{proof}

DR iteratively constructs a finite number of sets $\KC_\ell$.
In the proof of Theorem~\ref{thm: analysis of DR and DRS outputs}, we denote by $t$ the total number of these sets.
By construction, the sets $\KC_1, \ldots, \KC_t$ satisfy
\begin{align} \label{eq: nested sequence of K}
    [n] = \KC_1 \supset \KC_2 \supset \cdots \supset \KC_t = \KC,
\end{align}
where $n$ is the number of columns of $A$, and $\KC$ is the output produced by DR.
Furthermore, $\KC_\ell$ and $\KC_{\ell+1}$ satisfy
\begin{align} \label{eq: relation between K_l and K_(l+1)}
    \KC_{\ell+1} = \KC_{\ell} - i
\end{align}
where the element $i \in \KC_\ell$ is such that $\a_i \in \cone(A(\KC_\ell - i))$.
We are now ready to prove Theorem \ref{thm: analysis of DR and DRS outputs}.
\begin{proof}[(Proof of Theorem~\ref{thm: analysis of DR and DRS outputs})]
    We first prove the statements for the case of DR, and then for the case of DRS.

    \medskip \noindent
    \textbf{Case of DR.} \ We begin with statement (a).
    If $t=1$, then $[n] = \KC_1 = \KC_t = \KC$, and hence $\cone(A) = \cone(A(\KC))$, which establishes (a).
    Now assume $t \ge 2$. We show that $\cone(A(\KC_\ell)) = \cone(A(\KC_{\ell+1}))$ for each $\ell = 1, \ldots, t-1$,
    because this relation, together with $\KC_1 = [n]$ and $\KC_t = \KC$, implies statement (a).
    It suffices to prove the inclusion $\cone(A(\KC_\ell)) \subset \cone(A(\KC_{\ell+1}))$,
    since the reverse inclusion follows from relation~\eqref{eq: nested sequence of K}.
    Let $\a \in \cone(A(\KC_\ell))$.
    Then
    \begin{align*}
        \a = \sum_{k \in \KC_\ell} \alpha_k \a_k
    \end{align*}
    for some coefficients $\alpha_k \ge 0$.
    By relation~\eqref{eq: relation between K_l and K_(l+1)},
    we can rewrite this as
    \begin{align*}
        \a = \sum_{k \in \KC_{\ell + 1}} \alpha_k \a_k + \alpha_i \a_i
    \end{align*}
    where $i$ is the element removed from $\KC_\ell$ to form $\KC_{\ell+1}$.
    Since $\a_i \in \cone(A(\KC_{\ell} - i)) = \cone(A(\KC_{\ell + 1}))$ by relation~\eqref{eq: relation between K_l and K_(l+1)},
    we can express $\a_i = \sum_{k \in \KC_{\ell + 1}} \beta_k \a_k$ for some coefficients $\beta_k \ge 0$.
    Substituting this into the expression for $\a$ yields
    \begin{align*}
        \a = \sum_{k \in \KC_{\ell + 1}} \alpha_k \a_k + \alpha_i \sum_{k \in \KC_{\ell + 1}} \beta_k \a_k
        = \sum_{k \in \KC_{\ell + 1}} (\alpha_k + \alpha_i \beta_k) \a_k.
    \end{align*}
    Since each coefficient $\alpha_k + \alpha_i \beta_k$ is nonnegative,
    we conclude that $\a \in \cone(A(\KC_{\ell+1}))$, proving the desired inclusion.

    We next prove statement (b).
    Let $k \in \KC$.
    Then, by the construction of the sequence $\KC_1, \KC_2, \ldots, \KC_t$ generated by DR,
    there exists an index $\ell$ such that $k \in \KC_\ell$ and $\a_k \notin \cone(A(\KC_\ell - k))$.
    From relation~\eqref{eq: nested sequence of K}, we have  $\cone(A(\KC - k)) \subset \cone(A(\KC_{\ell} - k))$.
    Thus, $\a_k \notin \cone(A(\KC - k))$, even though $\a_k \in \cone(A)$.
    Hence, $\cone(A) \neq \cone(A(\KC - k))$ for any $k \in \KC$, which establishes statement (b).

    \medskip \noindent
    \textbf{Case of DRS.} \
    In this proof, we use the same notation as in the description of DRS (Algorithm~3.2), namely:
    \begin{itemize}
        \item the $p$-way partition $\{\IC_1, \ldots, \IC_p\}$ of $[n]$ produced in Step~1;
        \item the output $\JC_u$ produced by DR for the input matrix $A(\IC_u)$ in Step~2, for each $u \in [p]$;
        \item $\KC_u = \IC_u(\JC_u)$ for $u \in [p]$;
        \item $\IC = \KC_1 \cup \cdots \cup \KC_p$;
        \item the output $\JC$ produced by DR for the input matrix $A(\IC)$ in Step~3.
    \end{itemize}
    Note that $\KC$ is given by $\KC = \IC(\JC)$.

    We first prove statement (a).
    By statement (a) for the case of DR, we have $\cone(A(\IC_u)) = \cone(A(\KC_u))$ for each $u \in [p]$.
    Applying Lemma~\ref{lem: cone equality over partition} then yields
    \begin{align*}
        \cone(A) = \cone(A(\KC_1 \cup \cdots \cup \KC_p)) = \cone(A(\IC)).
    \end{align*}
    Furthermore, by statement (a) for DR, we obtain $\cone(A(\IC)) = \cone(A(\KC))$.
    Consequently, $\cone(A) = \cone(A(\KC))$, which proves statement (a) for DRS.

    We next prove statement (b).
    From statement (b) for the case of DR, we have $\JC - j \notin \Gamma(A(\IC))$ for any $j \in \JC$.
    This implies that
    \begin{align*}
        \cone(A(\IC(\JC) - i_j)) \neq \cone(A(\IC)),
    \end{align*}
    where $i_j$ denotes the $j$th element of $\IC = \{i_1, i_2, \ldots, i_m\}$  with $i_1 < \cdots < i_m$.
    Note that $i_j$ belongs to $\IC(\JC)$ because $j \in \JC$.
    Recall that $\KC = \IC(\JC)$.
    Thus, by letting $k = i_j$, we obtain
    \begin{align*}
        \cone(A(\IC(\JC) -  i_j)) = \cone(A(\KC - k)).
    \end{align*}
    Here, $k$ is an arbitrary element of $\KC$, because $j$ is arbitrary in $\JC$.
    Moreover, as shown in the proof of statement (a), we have $\cone(A) = \cone(A(\IC))$.
    Consequently, $\cone(A(\KC - k)) \neq \cone(A)$ for any $k \in \KC$.
    This proves statement (b) for the case of DRS.
\end{proof}

\section{Supplementary Experimental Results} \label{sec: supplementary experimental results}

\begin{figure}[h]
    \centering

    \begin{subfigure}[t]{0.48\linewidth}
        \centering
        \includegraphics[width=\linewidth]{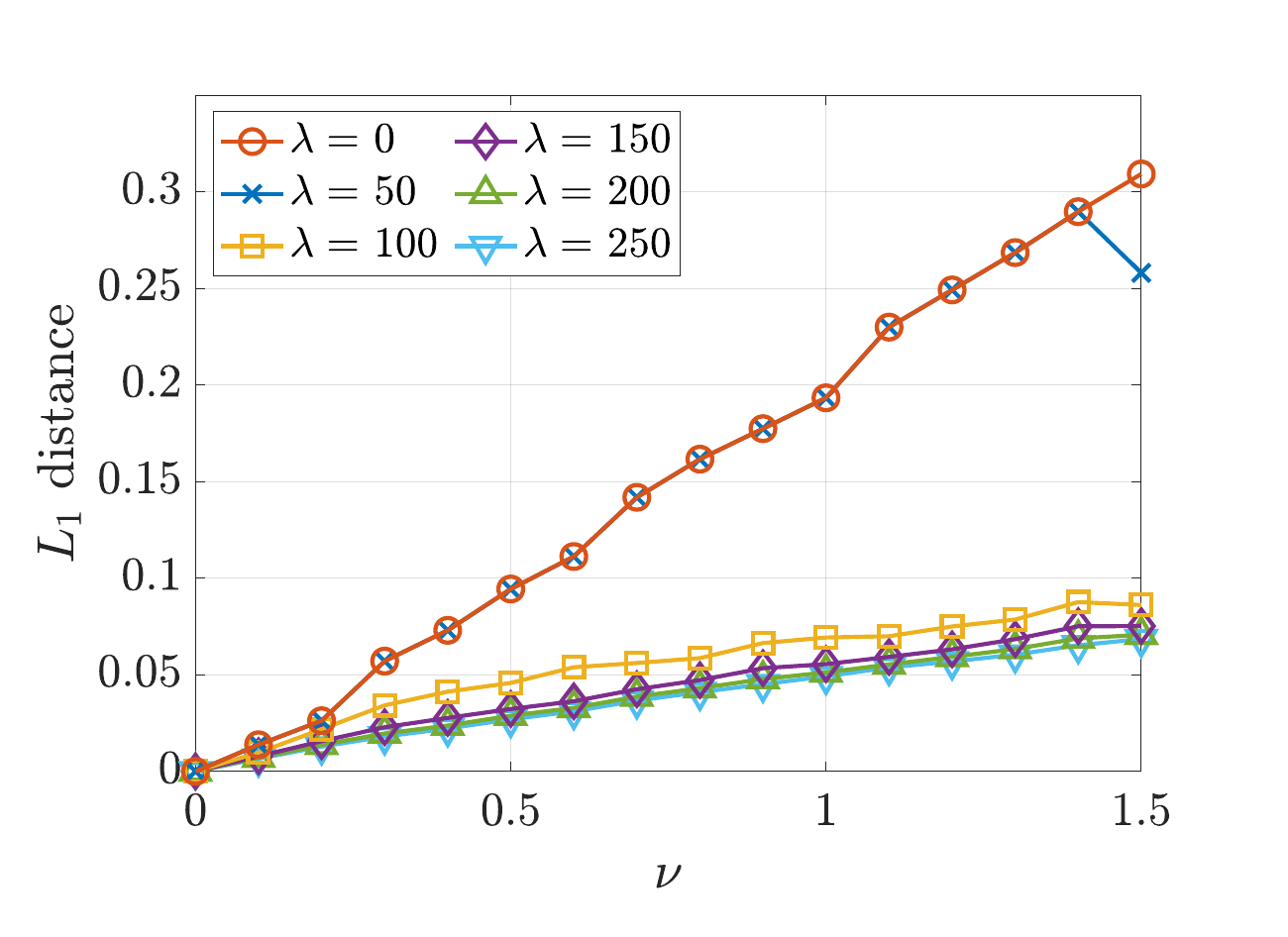}
        \caption{Dataset~1}
    \end{subfigure}
    \hfill
    \begin{subfigure}[t]{0.48\linewidth}
        \centering
        \includegraphics[width=\linewidth]{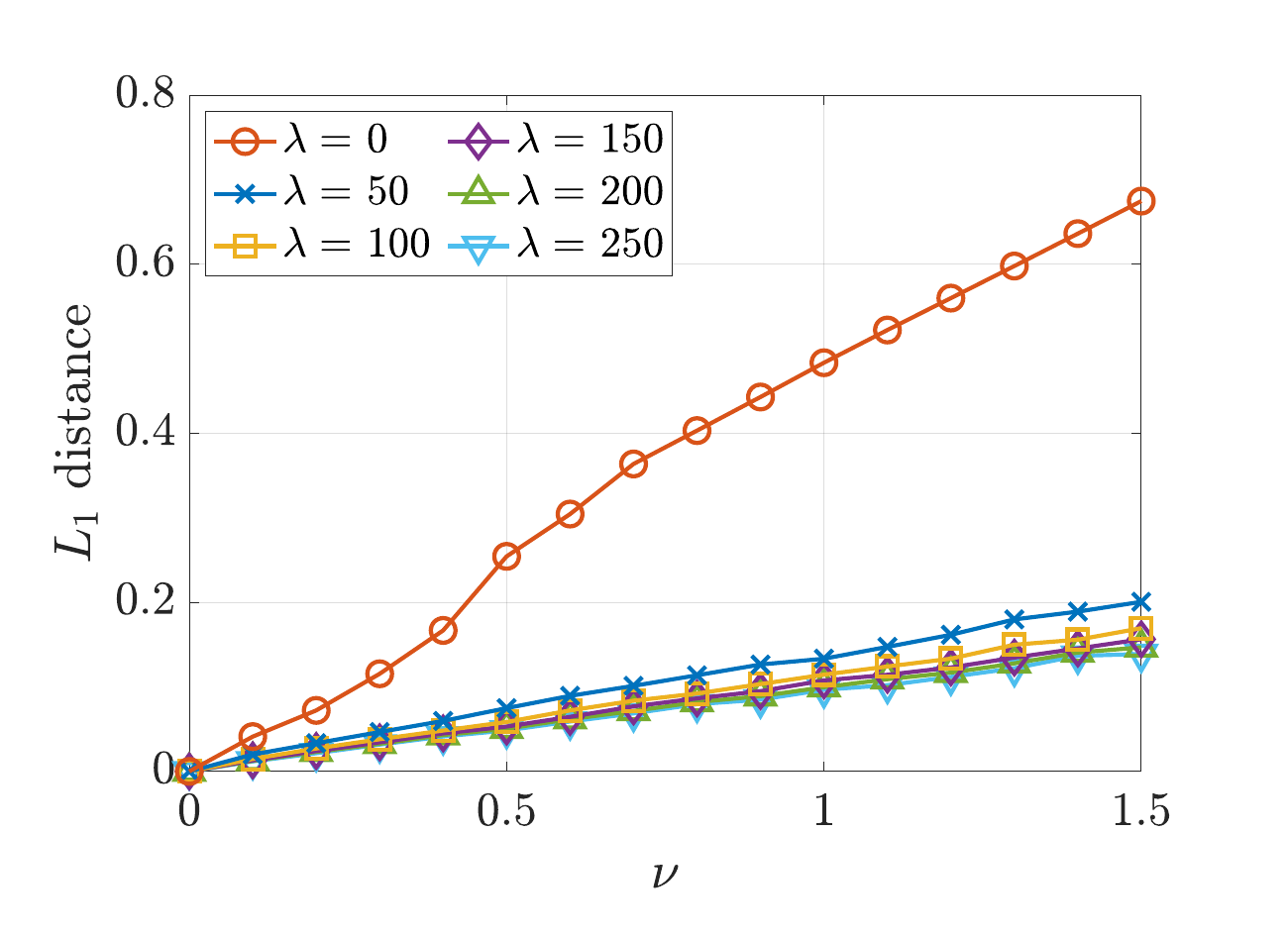}
        \caption{Dataset~2}
    \end{subfigure}

    \vspace{1mm}

    \begin{subfigure}[t]{0.48\linewidth}
        \centering
        \includegraphics[width=\linewidth]{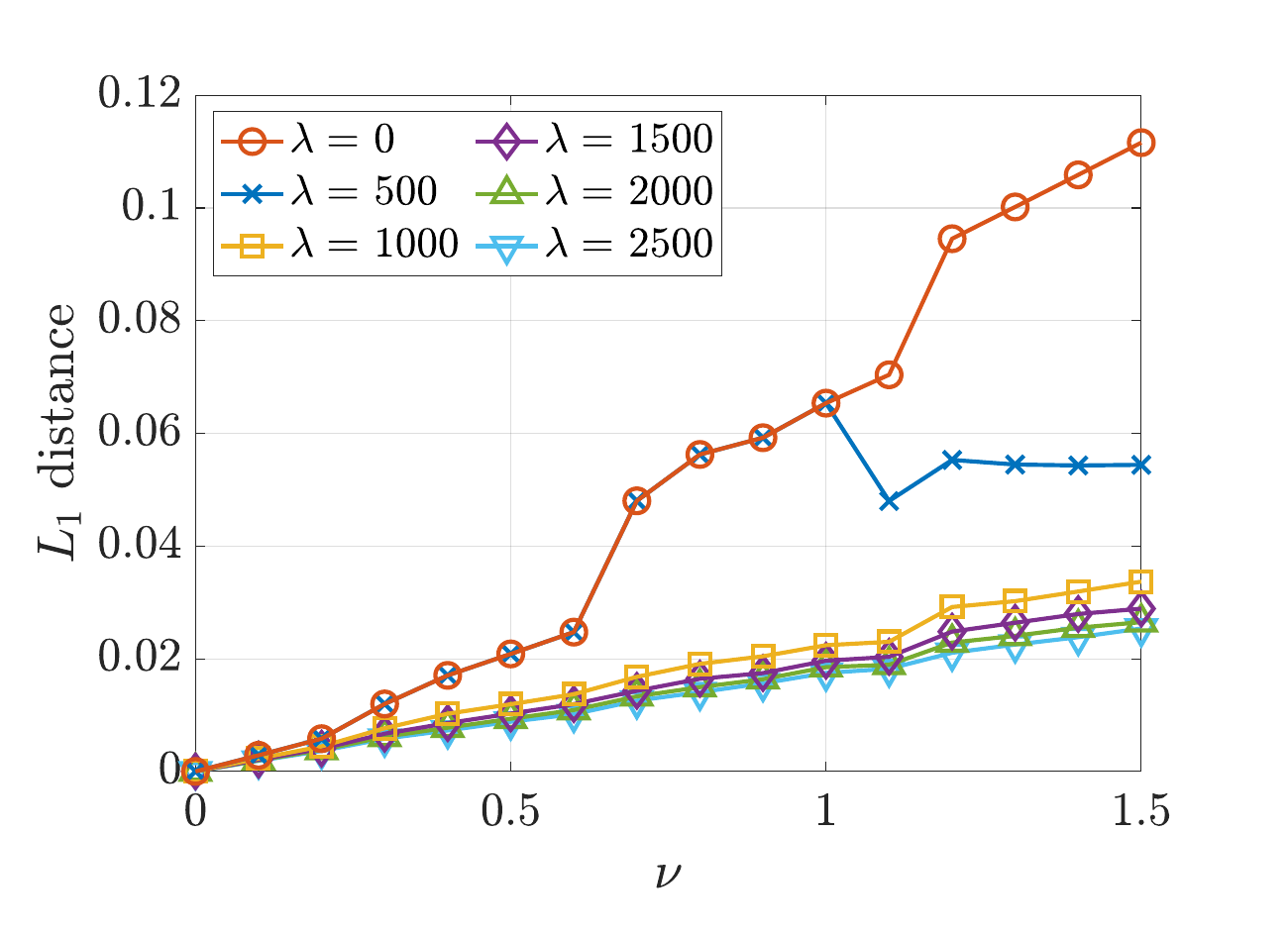}
        \caption{Dataset~3 with $\lambda$ from 0 to 2500 (step 500)}
    \end{subfigure}
    \hfill
    \begin{subfigure}[t]{0.48\linewidth}
        \centering
        \includegraphics[width=\linewidth]{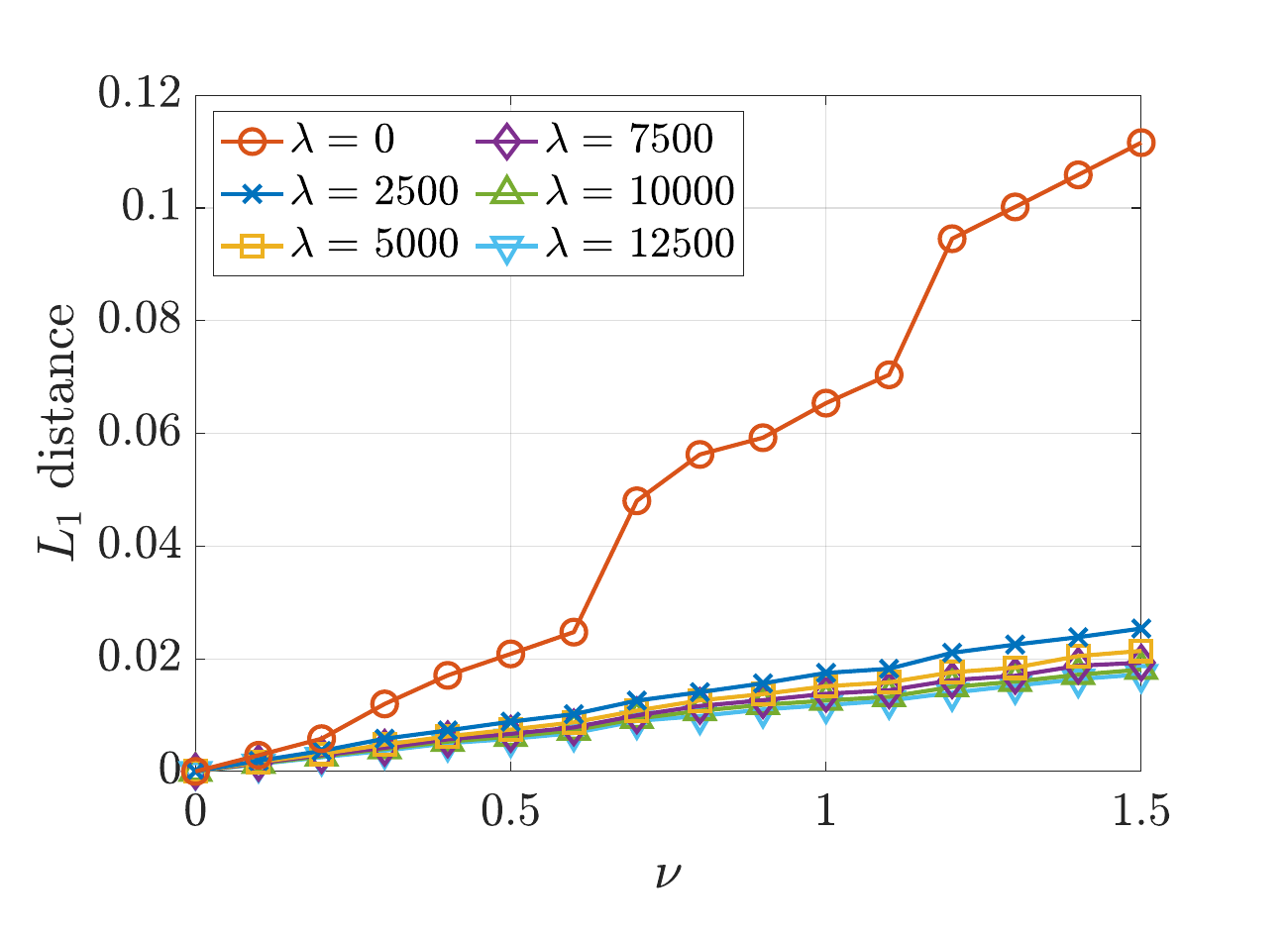}
        \caption{Dataset~3 with $\lambda$ from 0 to 12500 (step 2500)}
    \end{subfigure}

    \caption{Experimental results for Datasets~1--3 using the $L_1$ distance.
        Panels (a) and (b) show Datasets~1 and~2, while
        Panels (c) and (d) show Dataset~3 for different ranges of $\lambda$.
        Each panel shows the $L_1$ distance of the DRS output $\KC \cup \KC_{\add}$,
        averaged over 50 realizations of $\KC_{\add}$.}
    \label{fig: results of DRS for datasets 1-3 using the L1 distance}
\end{figure}

\section*{Acknowledgments}
The author would like to thank the anonymous reviewers for their helpful and insightful comments,
which significantly improved the quality of this paper.

\bibliographystyle{siamplain}
\bibliography{main.bbl}

\end{document}